\documentclass[10pt,reqno]{amsart}
\usepackage{amsmath}
\usepackage{latexsym}
\usepackage{amsfonts}
\usepackage{amssymb}
\usepackage{color}
\usepackage{bbm,dsfont}
\usepackage{graphicx}
\usepackage{hyperref}

\newtheorem{proposition?}{Proposition?}
\newtheorem{theorem}{Theorem}
\newtheorem{lemma}{Lemma}
\newtheorem{corollary}{Corollary}
\newtheorem{conjecture}{Conjecture}

\theoremstyle{definition}

\newtheorem{remark}{Remark}

\newtheorem{definition}{Definition}

\newcommand{\real}{\mathbb R} 
\newcommand{\complex}{\mathbb C} 
\newcommand{\nat}{\mathbb N} 

\newcommand{\hi}{\mathcal{H}} 
\newcommand{\hik}{\mathcal{K}} 
\newcommand{\lh}{\mathcal{L(H)}} 
\newcommand{\id}{\mathbbm{1}} 

\usepackage[left=2cm,top=2cm,right=2cm, bottom=1.95cm, a4paper]{geometry}

\begin{document}
\title[]{Equivalence classes and local asymptotic normality in system identification for quantum Markov chains}

\author[Guta]{Madalin Guta}
\email{madalin.guta@nottingham.ac.uk}
\author[Kiukas]{Jukka Kiukas}
\email{jukka.kiukas@nottingham.ac.uk}
\address{School of Mathematical Sciences, University of Nottingham, University Park,
Nottingham, NG7 2RD, UK}

\date{}
 
\begin{abstract} {
We consider the problems of identifying and estimating dynamical parameters of an ergodic quantum Markov chain, when only the stationary output is accessible for measurements. On the identifiability question, we show that the knowledge of the output state completely fixes the dynamics up to a `coordinate transformation' consisting of a multiplication by a phase and a unitary conjugation of the Kraus operators. 
%
When the dynamics depends on an unknown parameter, we show that the latter can be estimated at the `standard' rate $n^{-1/2}$, and give an explicit expression of the (asymptotic) quantum Fisher information of the output, which is proportional to the Markov variance of a certain `generator'. More generally, we show that the output is locally asymptotically normal, i.e. it can be approximated by a simple quantum Gaussian model consisting of a coherent state whose mean is related to the unknown parameter. As a consistency check we prove that a parameter related to the `coordinate transformation' unitaries, has zero quantum Fisher information.
}
\end{abstract}


\maketitle

\section{Introduction}
Quantum system identification has recently received significant attention due to its relevance in understanding complex quantum dynamical systems and the development of quantum technologies \cite{Dowling&Milburn}. Among the different  setups under investigation, we mention channel tomography \cite{Fujiwara}, Hamiltonian identification \cite{Burgarth,Cole}, and the estimation of the Lindblad generator of an open dynamical system  \cite{Howard}.  

In this paper we study two system identification problems set in a context which is particularly relevant for quantum control engineering applications \cite{Mabuchi&Khaneja,Wiseman&Milburn}, namely the input-output formalism of 
quantum Markov dynamics, which in the continuous time framework has been studied in \cite{Gardiner&Zoller}. Our setup is that of a discrete time quantum Markov chain consisting of a quantum system (or `memory') $\mathcal{H} = \mathbb{C}^D$ interacting successively with  input ancillas (or `noise units') 
$\mathcal{K}= \mathbb{C}^k$ which are identically prepared in a known pure state $|\chi\rangle\in \mathcal{K}$ 
(see Figure \ref{fig.markov}). The goal is to learn about the dynamics encoded in the unitary operator $U:\hi\otimes \hik\to \hi\otimes\hik$, by measuring the quantum output consisting of the noise units after the interaction. We will assume that we do not have direct access to the `memory' system, and we do not control the input state, both assumptions corresponding to realistic experimental constraints, e.g. in atom maser experiments \cite{HarocheRaimond}. We will also assume 
that the system's transition operator is irreducible and aperiodic, so that in the long time limit the dynamics has forgotten the initial state and  has reached stationarity.

\begin{figure}[h]
\includegraphics[width=10cm]{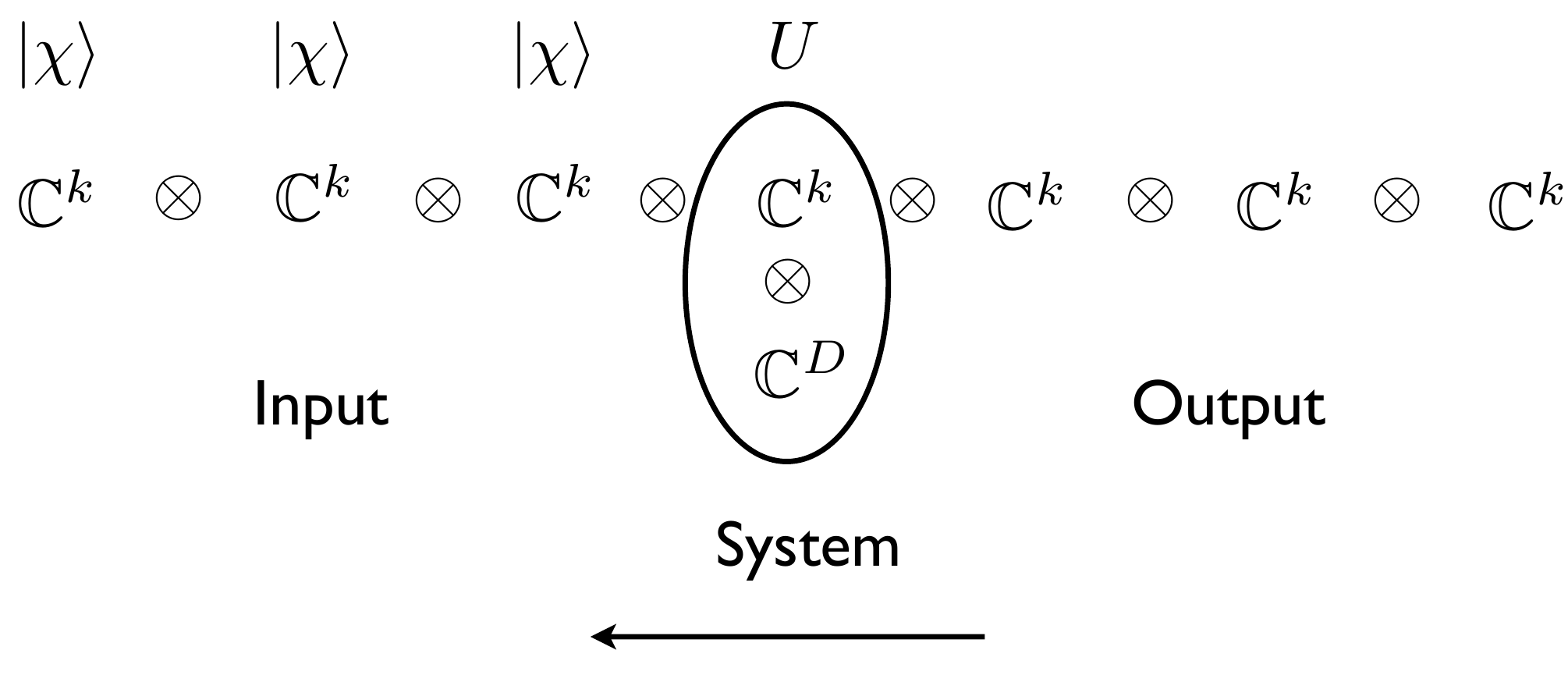}
\caption{Illustration of a quantum Markov chain. A sequence of identically prepared input ancillas (or `noise units') with space $\mathbb{C}^k$ interact successively (from right to left) with a system(or `memory') $\mathbb{C}^D$ via the unitary $U$. After $n$ interactions the output `noise units' are correlated with each other and with the system, and carry information about the unitary $U$.
}\label{fig.markov}
\end{figure}

In the stationary regime, the output is in a finitely correlated state \cite{FannesNachtergaeleWerner}  denoted $\rho^{out}_V$  which is completely determined by the isometry $V:\hi\to\hi\otimes\hik$ mapping the system into a system-ancilla space in one evolution step 
$$
V:|\varphi\rangle \mapsto  U (|\varphi\rangle \otimes |\chi\rangle).
$$
In particular, the restriction of the output state to the first $n$ noise units in the stationary regime is given by 
$\rho^{out}_V(n):={\rm tr}_{\rm \hi} [V(n)\rho_{ss} V(n)^*]$, where $\rho_{ss}$ is the system's stationary state and 
$V(n):\hi\to\hi\otimes \hik^{\otimes n}$ is the $n$-steps iteration of the isometry $V$ (see below for the precise definition). 

The first problem we address is that of finding the \emph{equivalence classes} of isometries with identical output states. In Theorem \ref{characterization} we show that $\rho^{out}_{V_1} = \rho^{out}_{V_2}$ if and only if $V_1$ and $V_2$ are related by an arbitrary complex phase $c$ and a conjugation with a local unitary $W$ on the system $\mathcal{H}$
$$
V_2=c (W^*\otimes \id_{\hik})V_1W.
$$   
We point out that a similar result holds in the following `classical' setup. 
A (finite) hidden Markov chain is a discrete time stochastic process $(X_n, Y_n)_{n\geq 0}$ consisting of an `underlying' Markov chain 
$(X_n)_{n\geq 0}$, with state space $\mathcal{X}:= \{1,\dots ,D\}$ and transition matrix $[T_{i,j}]_{i,i=1}^D$, and a sequence of random variables $(Y_n)_{n\geq 0}$ with values in $\mathcal{Y}:=\{1,\dots, k\}$, which `depend' on the underlying Markov dynamics. More precisely, $(Y_n)_{n\geq 0}$ are independent  when conditioned on $(X_n)_{n\geq 0}$, and $Y_n$ depends on the Markov chain only through $X_n$, with conditional distribution $R_{a,i} = \mathbb{P}(Y_n =a | X_n= i)$. We consider that $(X_n)_{n\geq 0}$ is `hidden' (not accessible to observations) and we observe the sequence $(Y_n)_{n\geq 0}$ whose marginal distribution depends only on the matrices $T$ and $R$ and the initial state $\pi$ of the Markov chain
$$
\mathbb{P}(Y_1= a_i, \dots ,Y_n= a_n)= \sum_{i_0,\dots, i_n} \pi(i_0) \prod_{j=0}^{n-1} T_{i_j, i_{j+1}} 
\prod_{j=0}^{n} R_{i_j, a_j}. 
$$
Under ergodicity conditions similar to ours, and certain additional generic conditions, Petrie \cite{Petrie} has shown that by observing $(Y_n)_{n\geq 0}$ in the stationary regime, we can identify the matrices $(T,R)$ up to a permutation of the labels of the hidden states in $\mathcal{X}$. Since the output process of a quantum Markov chain can be interpreted as a quantum analogue of a hidden Markov process, our result can be seen as a quantum extension of \cite{Petrie}.

After tackling the identifiability question, the second problem we consider is that of estimating the isometry $V$, or more precisely finding how the output state $\rho^{out}_V(n)$ varies with $V$, and how much statistical information it contains. Our approach is based on asymptotic statistics, and in particular on the concept of local asymptotic normality \cite{vanderVaart}. To illustrate this in the case of the  hidden Markov model, we assume that the matrices $(T,R)$ depend smoothly on an unknown real parameter $\theta$ which we would like to estimate, and denote by $\mathbb{P}^n_\theta$ the probability distribution of the data $(Y_1, \dots ,Y_n)$.  For large $n$ the unknown parameter can be `localised' in a neighbourhood of size $n^{-1/2}$  of a fixed value $\theta_0$ (e.g. a rough estimator), such that $\theta= \theta_0+ u/\sqrt{n}$ with $u$ a `local parameter' to be estimated. Local asymptotic normality means that the statistical model 
$\mathbb{P}^n_{\theta_0+u/\sqrt{n}} $ can be approximated by a simple Gaussian model consisting of a single sample from the normal distribution $N(u,I^{-1}_{\theta_0})$ with mean $u$ and fixed variance \cite{BickelRitovRyden2}. Closely related to this, and perhaps more relevant for the practitioner, is the fact that the maximum likelihood estimator 
$\hat{\theta}_n=\hat{\theta}_n(Y_1, \dots ,Y_n) $ is asymptotically normal \cite{Douc,BickelRitovRyden}, i.e.
$$
\sqrt{n} (\hat{\theta}_n-\theta) \overset{\mathcal{D}}{\longrightarrow} N(0,I^{-1}_{\theta_0})
$$ 
where the convergence holds in distribution, when $n\to\infty$. In particular, the mean square error 
$\mathbb{E}\left[(\hat{\theta}_n-\theta)^2\right]$ scales as $1/n $ with the best possible constant $I^{-1}_{\theta_0}$ which is the inverse limiting Fisher information per sample of the hidden Markov chain, at $\theta_0$.

Returning to our quantum system identification problem, we note that in order to find good estimators we need to deal with the optimisation problem of finding the `most informative' output measurement. In quantum state tomography this problem has been approached by developing quantum analogues of several key statistical notions such as quantum Fisher information \cite{Holevo,Helstrom}, and quantum local asymptotic normality \cite{KahnGuta,Guta&Jencova}. Equipped with these mathematical tools it is possible to solve the optimal state estimation problem in an asymptotic regime, by transforming it into a simpler Gaussian estimation one \cite{GillGuta}.  We will therefore apply the same strategy for the Markov system identification problem, but we note that since the output `noise units' are not independent, we cannot use the existing state tomography theory in a straightforward fashion. Instead, we developed a notion of convergence of \emph{pure states} statistical models based on the simple idea that two models are close to each other if the inner products of all corresponding pairs of vectors are very similar. To illustrate this, let $|\Psi_{u,\varphi}(n)\rangle: = V_{\theta_0 +u/\sqrt{n}}(n) |\varphi\rangle$ be the joint system and output state after $n$ steps, with initial condition $|\varphi\rangle\in \mathcal{H}$. In Theorem \ref{th.qlan} we show that 
$$
\lim_{n\to \infty} \left\langle \exp(-iau^2) \Psi_{u,\varphi}(n) |\exp(-iav^2) \Psi_{v,\varphi}(n)\right\rangle = \langle \sqrt{F/2} u | \sqrt{F/2}v\rangle =  \exp(- F(u-v)^2/8)
$$
where the right side is the inner product of two coherent states of the limit Gaussian model $|\sqrt{F/2}v\rangle$, with quantum Fisher information $F$. The latter can be computed explicitly and can be interpreted as a certain Markov variance of the `generator' $G= -iV\dot{V}^*$, similarly to the well known formula for unitary rotation families of states. Building on this result, we show that the convergence to the Gaussian model can be extended to the mixed stationary output state itself, and can be formulated in a strong, operational way by means of quantum channels connecting the models in both directions \cite{KahnGuta}. In this statistical picture, the local unitary conjugation corresponding to changes inside an equivalence class has quantum Fisher information equal to zero, in agreement with the fact that such parameters are not identifiable.

The paper is organised as follows. In section \ref{notations} we review the formalism of quantum Markov chains, discuss the ergodicity assumptions and introduce the `Markov covariance' inner product, which will be used later in interpreting the quantum Fisher information. In section \ref{sec.equivalence} we prove Theorem \ref{characterization} which characterises the equivalence classes of chains with identical outputs. Section \ref{sec.intermezzo} is a brief introduction to the statistical background of the paper, centred around the notion of convergence of quantum statistical models, and local asymptotic normality. The main result here is Lemma \ref{th.weak.implies.strong} which can be used to convert weak convergence of pure states models into strong convergence. Section \ref{sec.lan} contains several local asymptotic normality results for different versions of the output state. In all cases, the quantum model described by the output state can be approximated by a quantum Gaussian model consisting of a one parameter family of coherent states. 
Some of the more technical proofs are collected in section \ref{proof.LAN}.

A special case of the present local asymptotic normality result has been obtained in \cite{Guta}. A detailed study comparing the quantum Fisher information with classical Fisher informations of various counting statistics for the atom maser model can be found in \cite{CatanaGutavanHorssen,CatanaGutaKypraios}.  The continuous time version of the local asymptotic normality result will appear in a forthcoming publication \cite{Catana&Bouten&Guta}.
\section{Quantum Markov chains}\label{notations}

We begin by describing the general framework quantum Markov chains, and establishing some of the notations used in the paper. In the second part of this section we introduce a positive inner product describing the variance of Markov fluctuation operators, which will be used later for interpreting the limiting quantum Fisher information.


\subsection{Output state, Markov transition operator and the quantum Perron-Frobenius Theorem}\label{sec.Markov.dynamics}

A discrete time quantum Markov chain consists of a `system' with Hilbert space $\hi:= \complex^D$ which interacts successively 
with `noise units' (or ancillas, quantum coins) with identical Hilbert spaces $\hik:= \complex^k$, cf. Figure \ref{fig.markov}. The noise units are initially prepared independently, in the same state $|\chi\rangle\in \hik$, and the system has initial state $|\varphi\rangle\in \hi$.  
The interaction is described by a unitary operator $U:\hi\otimes \hik\to \hi\otimes\hik$, such that after one step, the joint state of system and the first unit  is $U(|\varphi\rangle\otimes |\chi\rangle)$, while the remaining units are still in the initial state. After the interaction with the first noise unit, the system moves one step to the \emph{left} and the same operation is repeated between system and the second unit, and so on. Alternatively, one can think that the system is fixed and the 
chain is shifted to the right. After $n$ steps the state of the system and the $n$ noise units is therefore
$$
|\Psi(n)\rangle:=  U(n)  |\varphi \otimes \chi^{\otimes n}\rangle \in \hi\otimes \hik^{\otimes n}
$$
where $U(n)$ is the product of unitaries $U^{(n)}\dots U^{(1)}$ with $U^{(i)}$ denoting the copy of $U$ which acts on system and the $i$-th noise unit, counting from \emph{right to left} according to the dynamics of Figure \ref{fig.markov}. 

Our goal is to investigate how the output state (reduced state of the noise units after the interaction) depends on the unitary $U$ and in particular which dynamical parameters can be identified by performing measurements on the output. With this in mind we note that the state $|\Psi(n)\rangle$ can be expressed as
$$
|\Psi(n)\rangle=  V(n)  |\varphi\rangle  
$$
where $V(1)\equiv V:\hi\to\hi\otimes\hik$ is the isometry given by $V|\varphi\rangle:= U|\varphi \otimes \chi\rangle$, and 
$V(n) :\hi\to\hi\otimes\hik^{\otimes n} $ is obtained iteratively from $V(k+1):=(V\otimes \id_{\hik^{\otimes k}})V(k)$. 
Therefore, since the input state is fixed and known, the output state depends on $U$ only through the isometry $V$ which will be the focus of our attention for most of the paper.

Let us fix orthonormal bases $\{|e_j \rangle\}_{j=1}^D$ and $\{|i\rangle\}_{i=1}^k$ in $\hi$ and respectively $\hik$, and let ${\bf K}:=\{K_i\}_{i=1}^k$ be the collection of Kraus operators acting on $\hi$, uniquely defined by the equation
\begin{equation}\label{VK}
V|\varphi\rangle = \sum_{i=1}^k K_i |\varphi\rangle\otimes |i\rangle.
\end{equation}
which satisfy the normalisation condition
\begin{equation}\label{normalisation}
\sum_{i=1}^k K_i^*K_i =\id
\end{equation}
Conversely, any set ${\bf K}:=\{K_i\}_{i=1}^k$ satisfying \eqref{normalisation} has a unique associated isometry $V:\hi\to \hi\otimes \hik$ given by \eqref{VK}. The system-output state can be written more explicitly in a matrix product form 
$$
|\Psi(n)\rangle= \sum_{{\bf i}\in I^{(m)}} {\bf K}_{\bf i}^{(m)}|\varphi\rangle \otimes |{\bf i}\rangle,
$$
where $I^{(m)}$ denotes the set of multi-indices ${\bf i}=(i_1,\ldots,i_m)$, and for each ${\bf i}\in I^{(m)}$ we denote by ${\bf K}_{\bf i}^{(m)}$ the $m$-step Kraus operator (note the backwards ordering)
$$
{\bf K}_{\bf i}^{(m)}:=K_{i_m}K_{i_{m-1}}\cdots K_{i_2}K_{i_1}.
$$
For later purposes  we define the concatenation of multi-indices ${\bf j}\in I^{(k)}$ and ${\bf i}\in I^{(m)}$ via $${\bf j}{\bf i}:= (i_1,\ldots, i_m,j_1,\ldots j_k)\in I^{(k+m)},$$
so that ${\bf K}_{\bf ji}^{(k+m)}={\bf K}_{\bf j}^{(k)}{\bf K}_{\bf i}^{(m)}$.

The properties of the system-output state $|\Psi(n)\rangle$ depend crucially on the irreducibility of the restricted system dynamics. 
Let $\mathcal{A}:= \mathcal{B}(\mathcal{H})$ and $\mathcal{B}:= \mathcal{B}(\mathcal{K})$ be the algebras of system and noise observables, and let us denote by 
$\mathcal{A}_*, \mathcal{B}_*$ their preduals (linear spans of density matrices). In the Schr\"{o}dinger picture the reduced evolution of the system is obtained by iterating the transition operator (quantum channel)
\begin{eqnarray*}
T_* &:& \mathcal{A}_* \to  \mathcal{A}_*\\
T_* &:&\rho                   \mapsto  V\rho V^* = \sum_{i=1}^k K_i \rho K_i^*
\end{eqnarray*}
as it can be seen from the identity
$$
V(n)  \rho V(n)^{*} = \sum_{{\bf i} \in I^{(n)}} {\bf K}_{\bf i}^{(n)}\rho {\bf K}_{\bf i}^{(n)*}=  T^n(\rho).
$$
The dual (Heisenberg) evolution is given by the unit preserving CP map 
\begin{eqnarray}
T&:& \mathcal{A} \to  \mathcal{A}\nonumber\\
T &:&X                   \mapsto  V^*X\otimes \id_{\hik}  V =  \sum_{i=1}^k K_i^* XK_i \label{eq.T&V}.
\end{eqnarray}

As for classical Markov chains, the system's space may possess non-trivial invariant subspaces, and have multiple stationary states. 
We will restrict our attention to such `building blocks' defined by the following properties.

\begin{definition}
A channel $T:\mathcal{A}\to \mathcal{A}$ is called 
\begin{itemize}
\item[(i)]
\emph{irreducible} if there exists $n\in \mathbb N$ such that $({\rm Id}+T)^n$ is strictly positive, 
i.e.  
$
({\rm Id}+T)^n(X) >0 
$
for all positive operators $X$. 

\item[(ii)] \emph{primitive}, if there exists an $n\in \mathbb N$, such that
$
T^n 
$
is strictly positive. 
\end{itemize}
We call an isometry $V:\hi\to\hi\otimes\hik$ irreducible (primitive), if the associated channel defined in \eqref{eq.T&V} is irreducible (primitive).
\end{definition}

Clearly primitivity is a stronger requirement than irreducibility. The following theorem (see \cite{Evans,Wolf}) collects the essential properties of irreducible (primitive) quantum transition operators needed in this paper.

\begin{theorem}[{\bf quantum Perron-Frobenius}]
Let  $T:\mathcal{A}\to \mathcal{A}$  be a completely positive map and let $r:= \max_i |\lambda_i| $ be its spectral radius, where 
$\{ \lambda_1, \dots ,\lambda_{D^2} \}$  are the (complex) eigenvalues of $T$ arranged in decreasing order of magnitude. Then 
\begin{itemize}
\item[(i)]
$r$ is an eigenvalues of $T$, and it has a positive eigenvector.

\vspace{1mm}

\item[(ii)]
If additionally, $T$ is unit preserving then $r=1$ with eigenvector $\id$.

\vspace{1mm}

\item[(iii)]
If additionally, $T$ is irreducible then $r$ is a non-degenerate eigenvalue for $T$ and $T_*$, and both corresponding eigenvectors are strictly positive.

\vspace{1mm}
 
\item[(iv)] If additionally, $T$ is primitive,  then $|\lambda_i| <r$ for all other eigenvalues than $r$.
\end{itemize}

\end{theorem}
We mainly need the following corollary of this theorem: If a channel $T$ is irreducible then it has a unique full rank stationary state, and if $T$ is also primitive then any state 
converges to the stationary state $\rho_{ss}$ in the long run (mixing or ergodicity property);
$$
\lim_{n\to\infty} (T_*)^n (\rho) = \rho_{ss},  
$$
or in the Heisenberg picture
\begin{equation}\label{statlimit}
\lim_{n\rightarrow \infty} T^n(X)= {\rm tr}[X\rho_{ss}]\id.
\end{equation}
Moreover, the speed of convergence is exponential, the rate depending on the second largest eigenvalue of $T$.
\subsection{The Markov covariance inner product}\label{sec.markov.covariance}
In this section we introduce an inner product playing the role of a `Markov covariance', whose relevance will be become apparent when computing the quantum Fisher information of the output. On a deeper level, we conjecture that the covariance is accompanied by a Central Limit Theorem, but this topic will not be pursued here (see \cite{Guta,Guta&vanHorssen} for the special case of output observables).

%

In this section we assume that $V$ is a primitive isometry, and we denote by $\langle Z \rangle^{in}_{ss} $ the expectation with respect to the input state $\rho_{ss} \otimes |\chi\rangle \langle\chi|^{\otimes{n}}$. Let $\mathcal{A}\otimes \mathcal{B}^{(j)}\subset\mathcal{A}\otimes \mathcal{B}^{\otimes n}$ be the subalgebra of observables of the 
system together with the $j$-th (from right to left) copy of the noise unit $\mathcal{B}$, and for any 
$X\in\mathcal{A}\otimes \mathcal{B} $ we denote by $X^{(j)}$ its version in $\mathcal{A}\otimes \mathcal{B}^{(j)}$. 
If $X\in\mathcal{A}\otimes \mathcal{B} $ we define the Heisenberg evolved operator after $i$ steps by
$$
X(i)= U^{(1)*} \dots U^{(i)*} X^{(i)} U^{(i)}\dots U^{(1)} \in \mathcal{A}\otimes\mathcal{B}^{\otimes n},
$$ 
and note that its stationary mean does not depend on $i$ and is given by
$$
\langle X(i)\rangle^{in}_{ss}=  
\langle \mathcal{E} (X) \rangle^{in}_{ss} 
= {\rm Tr( \rho_{ss} \mathcal{E}(X)})
$$
where $\mathcal{E}$ is the conditional expectation 
\begin{eqnarray*}
\mathcal{E}: \mathcal{A}\otimes \mathcal{B} &\to &\mathcal{A}\\
X&\mapsto &\langle \chi |U^* XU|\chi\rangle = V^* X V.
\end{eqnarray*} 
For all mean zero operators $ X\in\mathcal{A}\otimes\mathcal{B}$ we define the  associated \emph{fluctuations operator} by 
$$
\mathbb{F}_n(X) = \frac{1}{\sqrt{n}} \sum_{i=1}^n X(i)  \in \mathcal{A}\otimes \mathcal{B}^{\otimes n}.
$$
Finally, let $\mathcal{A}_0:= \{ A\in \mathcal{A}: {\rm Tr}[\rho_{ss}A]=0 \}$ and define $\mathcal{R}:\mathcal{A}_0\to \mathcal{A}_0$ to be the inverse of the restriction of the map ${\rm Id}- T$ to $\mathcal{A}_0$. Note that $\mathcal{R}$ is well defined, as it follows from the fact that $T$ is primitive and has a unique eigenvector $\id$ with eigenvalue $1$.


\begin{lemma}
Let $\mathcal{F}$ be the linear space of mean zero operators 
$\{ X\in \mathcal{A}\otimes \mathcal{B} : {\rm Tr}(\rho_{ss} \mathcal{E}(X))=0 \}$. For any  $X,Y\in\mathcal{F}$ and any $|\varphi\rangle\in \mathcal{K}$ the following limit exists 
$$
(X,Y)_V:=  \lim_{n\to\infty} \frac{1}{n}   
\left\langle \varphi \otimes \chi^{\otimes n}|  \mathbb{F}_n(X^*) \mathbb{F}_n(Y)  | \varphi \otimes \chi^{\otimes n}\right\rangle
$$
and defines a positive inner product on $\mathcal{F}$.

Moreover $(X,Y)_V$ has the explicit expression
$$
(X,Y)_V = {\rm Tr}\left\{ \rho_{ss} \mathcal{E}\left[
 X^*Y+ 
X^*  \left( \mathcal{R} \circ\mathcal{E}(Y)   \otimes \id \right) + 
\left(  \mathcal{R} \circ\mathcal{E}(X^*)  \otimes \id\right) Y  
\right]\right\}.
$$

\end{lemma}

\begin{proof}
For simplicity we denote by $\langle\cdot \rangle_{\varphi}^{in}$ the expectation with respect to the state 
$ | \varphi \otimes \chi^{\otimes n}\rangle$. By expanding the fluctuation operators we obtain
\begin{eqnarray*}
&&\frac{1}{n}  
\left\langle \varphi \otimes \chi^{\otimes n}|  \mathbb{F}_n(X^*) \mathbb{F}_n(Y)  | \varphi \otimes \chi^{\otimes n}\right\rangle =
\frac{1}{n}  
\sum_{i,j=1}^n 
\langle   X(i)^* Y(j) \rangle_{\varphi}^{in} \\
&=&\frac{1}{n} \sum_{i=1}^n \langle (X^*Y)(i)\rangle_{\varphi}^{in}+ 
\frac{1}{n}\sum_{i\neq j =1}^n \langle X(i)^*Y(j)\rangle_{\varphi}^{in}\\
&=& \frac{1}{n} \sum_{i=1}^n \langle \varphi | T^{i-1} \circ \mathcal{E}(X^*Y)| \varphi \rangle + 
\frac{1}{n}\sum_{1\leq i<j \leq n}^n \langle X(i)^*Y(j)\rangle_{\varphi}^{in}+  
\frac{1}{n}\sum_{1\leq j<i \leq n}^n \langle X(i)^*Y(j)\rangle_{\varphi}^{in}\\
&=& \frac{1}{n} \sum_{i=1}^n \langle \varphi | T^{i-1} \circ \mathcal{E}(X^*Y)| \varphi \rangle +
\frac{1}{n}\sum_{1\leq i<j \leq n}^n \langle \varphi  |T^{i-1}\circ\mathcal{E}\left[ X^*  \left( \id \otimes T^{j-i-1} \circ\mathcal{E}(Y) \right)\right] |
\varphi \rangle \\
&&+\frac{1}{n}\sum_{1\leq j<i \leq n}^n \langle \varphi  |T^{j-1}\circ\mathcal{E}\left[  \left( \id \otimes T^{i-j-1} \circ\mathcal{E}(X^*) \right) Y \right] |\varphi \rangle. \\
\end{eqnarray*}
Now, since $ T_*^{n} (|\varphi\rangle\langle \varphi|) $ converges to $\rho_{ss}$ exponentially fast as $n\to \infty$, the first term in the last equality converges to 
$ {\rm Tr}(\rho_{ss} \mathcal{E}(X^*Y))$. Similarly, the one of the summation indices in the second term can be changed to 
$k:= j-i-1$ such that it can be written as
$$
S_n:= \sum_{k=0}^{n-2} \frac{1}{n}\sum_{ i=1}^{n-k-1} \langle \varphi  |T^{i-1}\circ\mathcal{E}\left[ X^*  \left( \id \otimes T^{k} \circ\mathcal{E}(Y) \right)\right] |
\varphi \rangle.
$$
For any fixed $k$ the inner sum is dominated by terms with large $i$ and by the same stationarity argument as above, it 
converges to  
$$
{\rm Tr}\left( \rho_{ss} \mathcal{E}\left[ X^*  \left( \id \otimes T^{k} \circ\mathcal{E}(Y) \right)\right]\right).
$$
Now, by assumption, $\mathcal{E}(Y)$ belongs to the domain of $\mathcal{R}$ so that the sum converges
$$
\sum_{k\geq 0} T^k\circ\mathcal{E}(Y) = \mathcal{R}\circ\mathcal{E}(Y) .
$$
The two facts together imply that  
$$
\lim_{n\to\infty} S_n= {\rm Tr}\left\{ \rho_{ss} \mathcal{E}\left[
X^*  \left( \mathcal{R} \circ\mathcal{E}(Y)   \otimes \id \right) 
\right]\right\}.
$$

A similar reasoning applies to the third term of the sum.
\end{proof}

The previous lemma follows from the next conjecture which will not be investigated in this paper.  
\begin{conjecture}[Central Limit for quantum Markov chains]
Let $X\in \mathcal{A}\otimes \mathcal{B}$ be selfadjoint operator with $\langle X\rangle_{ss}^{out} =0$ The fluctuations operator 
$\mathbb{F}_n(X)$ satisfy the Central Limit Theorem
$$
\mathbb{F}_n(X) \overset{\mathcal{D}}{\longrightarrow} N(0, V_X)
$$ 
where $N(0, V_X)$ is the centred Gaussian distribution with variance $V_X = (X,X)_V$, and  the convergence holds as $n\to \infty$ in distribution with respect to the state $| \varphi \otimes \chi^{\otimes n}\rangle$ 
\end{conjecture}
The special case where $X\in \mathcal{B}$ has been proven in \cite{Guta}. Physically, it means that time averages of measurements of the observable $X$ are asymptotically Gaussian.   
\section{The equivalence class of quantum Markov chains with identical outputs}\label{sec.equivalence}

In Theorem \ref{characterization} of this section we answer the first question set in the introduction: which (mixing) quantum Markov chains have the same output states in the stationary regime? 
\begin{definition}\label{equiv_def} Let $V_l:\hi_l\to\hi_l\otimes\hik$, $l=1,2$ be two primitive isometries, where $\hi_l, \hik$ are finite-dimensional Hilbert spaces. Let $\rho_{ss,l}$ denote the respective stationary states. We call $V_1$ and $V_2$ \emph{equivalent}, if they have the same output states in the stationary regime, that is,
$$
{\rm tr}_{\hi_1}[V_1(n)\rho_{ss,1}V_1^*(n)]={\rm tr}_{\hi_2}[V_2(n)\rho_{ss,2} V_2^*(n)], \quad \text{ for all }n\in \nat.
$$
\end{definition}

We begin with a straightforward observation that for any given primitive isometry $V_1:\hi_1\to\hi_1\otimes\hik$, a number $c\in \complex$ with $|c|=1$, and a unitary $U:\hi_2\to\hi_1$, the isometry $V_2:=c(U^*\otimes I)V_1U$ is equivalent to $V_1$, and is primitive (with the stationary state $\rho_{ss,2}=U^*\rho_{ss,1}U$). The following Lemma characterises the case where two given primitive isometries are related this way. The proof is inspired by a similar argument from \cite{BaumgartnerNarnhofer}.

\begin{lemma}\label{firstlemma} Let $V_l:\hi_l\to\hi_l\otimes\hik$, $l=1,2$, be two primitive isometries, and define the maps
\begin{align*}
T_{ll'}:&\mathcal B(\hi_{l'},\hi_{l})\to \mathcal B(\hi_{l'},\hi_{l}), & T_{ll'}(X)&=V_l^*(X\otimes\id_{\mathcal K})V_{l'}.
\end{align*}
for $l,l'=1,2$. Then the following conditions are equivalent:
\begin{itemize}
\item[(i)] $T_{12}$ has an eigenvalue of modulus one;
\item[(ii)] $T_{21}$ has an eigenvalue of modulus one;
\item[(iii)] there exists a unitary operator $U:\hi_2\to\hi_1$, and $c\in \complex$ with $|c|=1$, such that
$$
V_2 = c(U^*\otimes\id_{\hik})V_1U.
$$
\end{itemize}
In that case, $T_{12}(U)=cU$ and $T_{21}(U^*)=\overline{c}U^*$.
\end{lemma}
\begin{proof} Conditions (i) and (ii) are clearly equivalent: if $T_{ll'}(F)=cF$ with some $c\in \complex$ and $F\in \mathcal B(\hi_{l'},\hi_{l})$, then $T_{l'l}(F^*)=\overline{c}F^*$.
Assuming (iii) we have
$$
T_{12}(U)= V_1^*(U\otimes \id_{\hik})V_2 = cV_1^*(U\otimes \id_{\hik})(U^*\otimes\id_{\hik})V_1U=cU,
$$
i.e. (i) holds, with $c$ the corresponding eigenvalue. Thus, the only nontrivial implication is (i)$\implies$ (iii).

Assume (i), let $c$ be an eigenvalue of $T_{12}$ modulus one, and $F$ such that $T_{12}(F)=cF$. Then from the definition of $T_{21}$ it follows immediately that $T_{21}(F^*)=\overline{c}F^*$. Since $V_1$ is an isometry, we have $V_1V_1^*\leq \id_{\hi_1\otimes\hik}$ (in fact, $V_1V_1^*$ is a projection). Hence,
\begin{align*}
T_{22}(F^*F)&=V_2^*(F^*\otimes \id_{\hik})(F\otimes \id_{\hik})V_2\\
&\geq V_2^*(F^*\otimes \id_{\hik})V_1V_1^*(F\otimes \id_{\hik})V_2=T_{21}(F^*)T_{12}(F) =|c|^2F^*F =F^*F,
\end{align*}
so by positivity of $T_{22}$,
$$
T_{22}^n(F^*F)\geq F^*F \quad \text{ for all } n\in \nat.
$$
Let $P$ be the projection onto the eigenspace of $F^*F$ corresponding to its largest eigenvalue $\|F^*F\|$. Now $\lim_{n\rightarrow\infty} T_{22}^n(X) ={\rm tr}[\rho_{ss,2} X]\id_{\hi_2}$ by the primitivity of $V_2$. Hence,
$$
{\rm tr}[\rho_{ss,2}F^*F]=\lim_{n\rightarrow\infty} {\rm tr}[P]^{-1}{\rm tr}[PT_{22}^n(F^*F)]\geq {\rm tr}[P]^{-1}{\rm tr}[PF^*F]=\|F^*F\|.
$$
This implies that ${\rm tr}[\rho_{ss,2}F^*F]=\|F^*F\|$, i.e. $\rho_{ss,2}$ is supported in the projection $P$. But $\rho_{ss,2}$ has full rank in $\hi_2$, so $P=\id_{\hi_2}$, and, consequently, $F^*F=\|F^*F\|\id_{\hi_2}$. By proceeding in exactly the same way using the primitive channel $T_{11}$, we show that $FF^*=\|FF^*\|\id_{\hi_1}$. Denote $\alpha:= \|FF^*\|=\|F^*F\|$, and $U:=\alpha^{-\frac 12} F$. Then $U:\hi_2\to\hi_1$ is a unitary operator between the two Hilbert spaces and in particular, $\dim\hi_1=\dim\hi_2$. Moreover,
\begin{equation*}
U^*V_1^*(U\otimes\id_{\hik})V_2=U^*T_{12}(U)=cU^*U = c\id_{\hi_2}.
\end{equation*}
For a unit vector $|\varphi\rangle\in \hi_2$, this implies
\begin{equation}\label{identity2}
\langle V_1U\varphi|(U\otimes\id_{\hik})V_2\varphi\rangle =c.
\end{equation}
But $\|V_1U\varphi\|=\|(U\otimes\id_{\hik})V_2\varphi\|=1$ by isometry, so there is actually \emph{equality} in the Cauchy-Schwartz inequality
$$
|\langle V_1U\varphi|(U\otimes\id_{\hik})V_2\varphi\rangle|\leq \|V_1U\varphi\|\|(U\otimes\id_{\hik})V_2\varphi\|.
$$
This is possible only if the two vectors are linearly dependent, i.e. there is a constant $\lambda\in \mathbb C$ such that
$$
|(U\otimes\id_{\hik})V_2\varphi\rangle=\lambda|V_1U\varphi\rangle.
$$
Putting this back in \eqref{identity2}, we see that $\lambda=c$. Since $\varphi$ was arbitrary, we have (ii), and the proof is complete.
\end{proof}

The following lemma deals with the case where all eigenvalues of the map $T_{12}$ of the preceding lemma have modulus strictly less than one.

\begin{lemma}\label{secondlemma} Let $V_l:\hi_l\to\hi_l\otimes\hik$, $l=1,2$, be two primitive isometries, and define $T_{ll'}$ as in Lemma \ref{firstlemma}. Let $\rho_{l}^{out}(n)$ be the associated output states.
Then the limits
$$
\lim_{n\rightarrow\infty} {\rm tr}[\rho_{1}^{out}(n)^2],\quad {\rm and}\quad \lim_{n\rightarrow\infty} {\rm tr}[\rho_{2}^{out}(n)^2]
$$
exist and are strictly positive. If
$\lim_{n\rightarrow\infty} T_{12}^n=0$, or, equivalently,
$\lim_{n\rightarrow\infty} T_{21}^n=0$, then
\begin{equation}\label{asymptotic_ortho}
\lim_{n\rightarrow\infty} {\rm tr}[\rho_{1}^{out}(n)\rho_{2}^{out}(n)]=0.
\end{equation}
\end{lemma}
\begin{proof}
Let $\rho_{ss,l}$ be the stationary state of $T_{ll}$, with spectral decomposition
$$
\rho_{ss,l} = \sum_{i=1}^{\dim\hi_l} \Lambda_{l,i} |e_{l,i}\rangle\langle e_{l,i}|, \qquad l=1,2.
$$
The output states decompose as follows:
\begin{align*}
\rho_{l}^{out}(n)&=\sum_{i=1}^{\dim\hi_l} \Lambda_{l,i} {\rm tr}_{\hi_l}[|V_l (n)e_{l,i}\rangle\langle V_{l}(n)e_{l,i}|]\\
&=\sum_{i,j=1}^{\dim \hi_l} \Lambda_{l,i} |\psi_{l,ji}(n) \rangle\langle \psi_{l,ji}(n)|,
\end{align*}
where
$$
\psi_{l,ji}(n) =\sum_{\bf i\in I^{(n)}} \langle e_{l,j}|K^{(n)}_{l,\bf i}e_{l,i}\rangle |{\bf i}\rangle.
$$
Now
\begin{align*}
\langle \psi_{l,ji}(n)|\psi_{l',j'i'}(n)\rangle &=
\sum_{{\bf i}\in I^{(n)}} \overline{\langle e_{l,j}|{\bf K}_{l,\bf i}^{(n)}e_{l,i}\rangle}\langle e_{l',j'}|{\bf K}_{l,\bf i}^{(n)} e_{l',i'}\rangle\\
&=\langle e_{l,i}|T_{ll'}^n(|e_{l,j}\rangle\langle e_{l'j'}|)|e_{l',i'}\rangle,
\end{align*}
so we can write
\begin{align*}
{\rm tr}[\rho_{l}^{out}(n)\rho_{l'}^{out} (n)]&= \sum_{i,j=1}^{\dim \hi_l}\sum_{i',j'=1}^{\dim \hi_{l'}} \Lambda_{l,i}\Lambda_{l',i'} |\langle\psi_{l,ji}(n)|\psi_{l',ji}(n)\rangle|^2\\
&=
\sum_{i,j=1}^{\dim \hi_l}\sum_{i',j'=1}^{\dim \hi_{l'}} \Lambda_{l,i}\Lambda_{l',i'} |\langle e_{l,i}|T_{ll'}^n(|e_{l,j}\rangle\langle e_{l'j'}|)|e_{l',i'}\rangle|^2.
\end{align*}
Since the isometries $V_l$ are primitive, we have $T_{ll}^n(X)\rightarrow {\rm tr}[\rho_{ss,l}X]\id_{\hi_l}$ for any $X\in \mathcal B(\hi_l)$ by \eqref{statlimit}. Now, on the one hand, by choosing $l=l'$ we get
$$
\lim_{n\rightarrow\infty}{\rm tr}[\rho_{l}^{out}(n)^2]=\sum_{i,i'=1}^{\dim\hi_l}\Lambda_{l,i}^2\Lambda_{l,i'}^2>0.
$$
On the other hand, assuming $\lim_{n\rightarrow\infty} T_{ll'}=0$ for $l\neq l'$, we get \eqref{asymptotic_ortho}.
\end{proof}

We are now ready to prove the main result of this section.

\begin{theorem}\label{characterization}
Two primitive isometries $V_l:\hi_l\to\hi_l\otimes\hik$, $l=1,2$, are equivalent if and only if there exists a unitary operator $U:\hi_2\to\hi_1$, and a complex number $c$ with $|c|=1$, such that
\begin{equation}\label{unitary_conj}
V_2=c (U^*\otimes \id_{\hik})V_1U.
\end{equation}
\end{theorem}
\begin{proof}
As mentioned above, the `if' part is straightforward. Assume now that $V_1$ and $V_2$ are equivalent, and define $T_{ll'}$ as in Lemma \ref{firstlemma}. We consider the direct sum isometry
$$V_{\rm tot}:=V_1\oplus V_2:\hi_1\oplus\hi_2\to \hi_1\otimes\hik\oplus\hi_2\otimes\hik=(\hi_1\oplus\hi_2)\otimes \hik.$$
We identify the elements $X\in\mathcal B(\hi_1\oplus\hi_2)$ in the usual way with block matrices
$$
X=\begin{pmatrix} X_{11} & X_{12}\\ X_{21} & X_{22}\end{pmatrix},
$$
where $X_{ll'}\in \mathcal B(\hi_{l'},\hi_{l})$, the set of linear operators $\hi_{l'}\to\hi_{l}$. This identifies $\mathcal B(\hi_l,\hi_{l'})$ as a subspace $\mathcal B(\hi_1\oplus\hi_2)$, and each of these four subspaces is invariant under the channel $T$ associated with $V_{\rm tot}$. Explicitly, we have
\begin{equation}\label{phirep}
T(X) =\begin{pmatrix} T_{11}(X_{11}) & T_{12}(X_{12})\\ T_{21}(X_{21}) & T_{22}(X_{22})\end{pmatrix}.
\end{equation}
In particular, any eigenvalue of $T_{ll'}$ is also an eigenvalue of $T$, because the subspaces are invariant. Since $T$ is completely positive and unital by construction, all eigenvalues of $T_{ll'}$ have modulus at most one. If all eigenvalues of $T_{12}$ have modulus strictly less than one, then $\lim_{n\rightarrow\infty} T_{12}^n=0$, which according to Lemma \ref{secondlemma} contradicts the assumption that the output states are equal. Hence $T_{12}$ has an eigenvalue of modulus one, so Lemma \ref{firstlemma} concludes the proof.
\end{proof}
\begin{remark} According to \cite{Wolf}, and applying a reconstruction procedure analogous to the one given in \cite{Gross}, one observes that the finitely correlated output states $\rho^{out}(n)$, $n\in \mathbb N$, associated to a primitive isometry $V$ are completely determined by one state $\rho^{out}(n_0)$, provided that $n_0\geq 2(D^2-d+1)D^2$, where $d$ is the number of linearly independent Kraus operators. Hence, for any two isometries $V_l$, $l=1,2$, there exists a finite $n_0$ such that $\rho^{out}_1(n_0)=\rho_2^{out}(n_0)$ if and only if \eqref{unitary_conj} holds for some phase factor $c$ and unitary $U$.
\end{remark}

\section{Intermezzo on convergence of quantum statistical models and local asymptotic normality}
\label{sec.intermezzo}

In this section we introduce the basic elements of a theory of quantum statistical models, in as much as it is necessary to understand 
the second main result presented in the next section: the local asymptotic normality of the quantum Markov chain's output, and the associated quantum Fisher information. This section is not directly connected to the Markov set-up and can be skipped at a first reading.

\begin{definition}
Let $\Theta$ be a parameter space. A quantum statistical model over $\Theta$ is a family 
$$
 \mathcal{Q}:= \{\rho_\theta \, :\, \theta\in \Theta\},
$$ 
of density matrices $\rho_\theta$ on a Hilbert space $\hi$, which are indexed by an unknown parameter $\theta\in \Theta$.
\end{definition} 

The typical quantum statistical problem associated to a model $\mathcal{Q}$ is to estimate the unknown parameter $\theta$ by measuring a system prepared in the state $\rho_\theta$, and constructing an estimator $\hat{\theta}$ based on the measurement outcome. In practice, the problem typically involves an additional parameter $n$ describing the `sample size', and `good estimators' have the property that the estimation error (e.g. the mean square error $\mathbb{E} (( \hat{\theta}_n-\theta)^2)$) converges to zero as $n\to\infty$. The samples may be independent and identical as in quantum state tomography, or may consist of correlated systems depending on an unknown dynamical parameter, as considered in this paper. The \emph{rate} of convergence is typically of the order $n^{-1}$ and has a constant factor equal to the inverse of the (asymptotic) Fisher information, the latter describing the amount of statistical information \emph{per sample}.

Asymptotic statistics deals with the `large n' statistical inference set-up. The power of this set-up lies in the fact that one can take advantage of general Central Limit behaviour, and approximate the `n samples' statistical model by simpler Gaussian models, with asymptotically vanishing approximation error \cite{vanderVaart}.

\subsection{I.I.D. pure state models.}\label{sec.i.i.d.lan}
 To illustrate this idea, let us consider a simple model consisting of $n$ qubits which are independent and identically prepared in a pure state depending on a two-dimensional rotation parameter
 $$
|\psi_\theta\rangle := \exp\left(i \sqrt{2}(\theta_2\sigma_{x} - \theta_1\sigma_{y})\right) | 0\rangle, \qquad
 \theta= (\theta_1, \theta_2)\in \mathbb{R}^2
$$
where the factor $\sqrt{2}$ has been inserted for later convenience.
Since we work in an asymptotic framework, we can consider that the parameter $\theta$ belongs to a neighbourhood of size $n^{-1/2 +\epsilon}$ of a \emph{known} fixed value $\theta_0$ which by symmetry can be chosen to be $\theta_0=(0,0)$. Such a `localisation' is not a prior assumption, but can be achieved with an adaptive procedure where a `small' sample 
$\tilde{n}= n^{1-\epsilon}\ll n$ is used to produce a rough estimate $\theta_0$, and this information is fed into the design of the second stage optimal measurement \cite{KahnGuta}. We will therefore write $\theta= u/\sqrt{n} =( u_1/\sqrt{n}, u_2/\sqrt{n})$, where $u= (u_1,u_2)$ is a local parameter to be estimated. In this case, the Gaussian approximation mentioned above is closely related to what is known in physics as the Holstein-Primakov approximation for coherent spin states \cite{HolsteinPrimakov}. By a Central Limit argument one can show that in the limit of large $n$ the collective spin variables 
$$
L_x := \frac{1}{\sqrt{2n}} \sum_{i=1}^n \sigma_{x}^{(i)}, \qquad
L_y := \frac{1}{\sqrt{2n}} \sum_{i=1}^n \sigma_{y}^{(i)}
$$
converge in joint moments (with respect to the product state $|\psi_{u/\sqrt{n}}\rangle^{\otimes n} $)  to continuous variables 
$Q$ and respectively $P$ which satisfy the canonical commutation relations $[Q,P] = i\id$ and have a 
coherent (Gaussian) state $|u\rangle $ with means 
$
\langle Q\rangle = u_1
$
and 
$
\langle P\rangle = u_2.
$
Therefore, in the large $n$ limit the i.i.d. qubit model 
$$
\mathcal{Q}_n:= \left\{ |\psi_{u/\sqrt{n}}\rangle^{\otimes n} :  u\in \mathbb{R}^2  \right\}
$$ 
is approximated (locally around $\theta_0$) by the `quantum Gaussian shift' model
$$
\mathcal{G} := \left\{|u\rangle: u\in \mathbb{R}^2 \right\}.
$$ 
From the statistical viewpoint this approximation (when formulated in an appropriate way) provides the asymptotically optimal measurement procedures and estimation rates \cite{Guta&Kahn,KahnGuta}. Indeed if we would like to estimate $u_1$, then the optimal measurement is that of the total spin $L_x$ which corresponds to the canonical variable $Q$ in the limit model, and similarly for $u_2$. However if one is interested in both parameters (i.e. the mean square error is 
$\mathbb{E} (\|\theta-\hat{\theta}\|^2 )$) then the optimal procedure is to measure each $L_x$ and $L_y$ on half of the 
spins, which corresponds to the heterodyne measurement for the Gaussian model, where the coherent state is split in 
two, and conjugate canonical variables are measured separately on the two subsystems. 

The above convergence to the Gaussian model can be captured in a simple way, as convergence of the inner products for arbitrary pairs of local parameters
$$
\lim_{n\to \infty} 
\left\langle\psi_{u /\sqrt{n}}^{\otimes n} \right|\left.  \psi_{v /\sqrt{n}}^{\otimes n}\right\rangle= \langle u|v\rangle, \qquad u,v\in \mathbb{R}^2.
$$
This `weak convergence' \cite{Guta&Jencova,Guta} has an appealing geometric interpretation and can be verified more easily than the  `strong convergence' investigated in \cite{Guta&Kahn,KahnGuta}, the latter being nevertheless  more powerful and applicable to more general models of mixed states. At the end of the section we will show how weak convergence can be upgraded to strong convergence under an additional assumption.
%

\subsection{Weak and strong convergence of pure state models.} 

We will now briefly sketch a mathematical framework for weak convergence of pure state models, which will serve as motivation for our results on local asymptotic normality for quantum Markov chains. Since this is not the main focus of the paper, we leave the general theory for a separate work.

\begin{definition}\label{def.equivalence}
Let $\mathcal{Q}: =\{ \rho_\theta \,:\, \theta\in \Theta \}$ be a quantum model with parameter space $\Theta$ and Hilbert space $\mathcal{H}$, and let  $\mathcal{Q}^\prime: =\{ \rho_\theta^\prime \,:\, \theta\in \Theta\}$ be another model with the same parameter space and Hilbert space $\mathcal{H}^\prime$. We say that $\mathcal{Q}$ is equivalent to 
$\mathcal{Q}^\prime$ if there exists quantum channels $T$ and $S$ such that 
$$
T(\rho_\theta) = \rho_\theta^\prime, \qquad S(\rho_\theta^\prime) = \rho_\theta, \quad \forall \,\theta\in \Theta.
$$ 
\end{definition}

It can be easily seen that if two models are equivalent then for any statistical decision problem, their optimal procedures can be related through the channels $T$ and $S$ and the corresponding risks (figures of merit) are equal \cite{Guta&Jencova}. The definition is also naturally connected with the theory of quantum sufficiency developed in \cite{Petz&Jencova}. We will now extend this to allow for models which are `close' to each other but not necessarily equivalent. For our purposes, it suffices to restrict our attention to pure state models.

 Let $\mathcal{Q}$
 and 
 $\mathcal{Q}^\prime$
be as in definition \ref{def.equivalence}, with $ \rho_\theta =|\psi_\theta\rangle\langle \psi_\theta| $ and 
$\rho_\theta^\prime =|\psi_\theta^\prime\rangle\langle \psi_\theta^\prime |$ pure state on $\hi$ and respectively $\hi^\prime$. 
It has been shown \cite{Chefles} that the two models are equivalent if and only there exists a choice of phases such that the inner products coincide 
$$
\langle \psi_{\theta_1} |\psi_{\theta_2}\rangle= \langle \psi_{\theta_1}^\prime |\psi_{\theta_2}^\prime 
\rangle, 
$$ 
 for all $\theta_1,\theta_2\in \Theta$. 
This suggests the following definition of convergence of pure state statistical models.
\begin{definition}\label{def.weak.convergence}
A sequence of quantum statistical models 
$
\mathcal{Q}_{n}:=
\{ |\psi_{\theta}(n)\rangle \langle \psi_{\theta}(n)|  \,:\,  \theta\in\Theta\}$ over spaces $\mathcal{H}_{n}$ 
\emph{converges weakly} to a model
 $
\mathcal{Q}:= 
\{ |\psi_{\theta}\rangle\langle \psi_{\theta} | \,:\,  \theta\in\Theta\}
$
over the space $\mathcal{H}$, if there exists a choice of phases for $|\psi_{\theta}(n)\rangle$ such that 
\begin{equation}\label{eq.weak.convergence}
\lim_{n\to\infty} 
\langle \psi_{\theta_{1}}(n)| \psi_{\theta_{2}}(n) \rangle
=
\langle \psi_{\theta_{1}}| \psi_{\theta_{2}}\rangle
\end{equation}
for all $\theta_{1},\theta_{2} \in\Theta$.

\end{definition}
Without entering into details we mention that this notion of convergence is closely related with that defined in \cite{Guta&Jencova}, which in turn is a quantum analogue of the classical weak convergence of statistical models \cite{vanderVaart}.

We introduce now a second notion of convergence of models which has a clearer operational interpretation and can be used to devise asymptotically optimal estimation procedures, and establish the asymptotic normality of estimators, in both the classical \cite{vanderVaart} and the quantum \cite{KahnGuta} contexts. For this, we will use definition \ref{def.equivalence} to build a distance between models, such that equivalent models have distance equal to zero, and small distance means that the models are `statistically close' to each other but not necessarily equivalent.
\begin{definition}\label{def.strong.convergence}
With the notation of definition \ref{def.equivalence}, we define the \emph{Le Cam distance} between the models $\mathcal{Q}$ and $\mathcal{Q}^\prime$ by $\Delta(\mathcal{Q}, \mathcal{Q}^\prime)= \max (\delta(\mathcal{Q}, \mathcal{Q}^\prime) ,\delta(\mathcal{Q}^\prime, \mathcal{Q}) )$ where 
$$
\delta(\mathcal{Q}, \mathcal{Q}^\prime)  = \inf_T \sup_{\theta\in \Theta} \| T ( \rho_\theta) - \rho^\prime_\theta\|_1, \qquad
\delta(\mathcal{Q}^\prime, \mathcal{Q})  = \inf_S \sup_{\theta\in \Theta} \| S ( \rho^\prime_\theta) - \rho_\theta\|_1,
$$
with $T$ and $S$ arbitrary quantum channels between the appropriated spaces. 

A sequence of models $\mathcal{Q}_n:=\{\rho_\theta(n) \,:\, \theta\in \Theta \}$ converges strongly to 
$\mathcal Q$ if $\Delta(\mathcal{Q}, \mathcal{Q}_n) \rightarrow 0$ as $n\rightarrow \infty$.
\end{definition}
It is easy to see that the strong convergence of a sequence $\mathcal{Q}_n$ to $\mathcal Q$ is equivalent to the existence of sequences $(T_n)$ and $(S_n)$ of channels such that
$$
\lim_{n\rightarrow\infty}\sup_{\theta\in \Theta} \| T_n ( \rho_\theta(n)) - \rho_\theta\|_1=0, \qquad \text{and}\quad
\lim_{n\rightarrow\infty}\sup_{\theta\in \Theta} \| S_n ( \rho_\theta) - \rho_\theta(n)\|_1=0.
$$
Indeed, assuming strong convergence, we find inductively for each $k\in \mathbb N$ an $n_k\geq n_{k-1}$ such that for each $n\geq n_k$ there exists a channel $T^k_n$ with $\sup_{\theta\in \Theta} \| T^k_n ( \rho_\theta(n)) - \rho_\theta\|_1<1/k$. Then the first of the above limits holds, if we define $T_n$ for each $n\in \mathbb N$ by putting $T_n:=T_n^k$ for that $k$ for which $n\in [n_k, n_{k+1})$. The other relation follows similarly. The converse implication is clear.

The interpretation of the strong convergence of (quantum) statistical models is that asymptotically, the optimal risk associated to a statistical problem (e.g. estimation) for $\mathcal{Q}_n$ converges to the optimal risk of the limit model; additionally, the optimal procedures can be related asymptotically via the channels  $T_n$ and $S_n$. A particular instance of strong convergence is \emph{local asymptotic normality} of i.i.d. quantum models for mixed finite dimensional states \cite{KahnGuta}, which shows that for large  sample size the ensemble of identically prepared systems can be approximated via quantum channels by a multi-mode classical-quantum 
Gaussian state of unknown mean.  

In this paper we establish a similar result for the correlated model given by the output state of a Markov chain. As a corollary we obtain the limiting (asymptotic) quantum Fisher information of the Markov output. This is done by proving the simpler weak convergence, which can be converted into strong convergence by applying Lemma \ref{th.weak.implies.strong} below. In preparation we prove a simpler lemma for finite number of parameters.

\begin{lemma}\label{lemma.weak.stronng.equivalence}
Let $\mathcal{Q}_{n}$ and $\mathcal{Q}$ be as in definition \ref{def.weak.convergence}, and suppose that $\mathcal{Q}_{n}$ converges weakly to $\mathcal{Q}$. Assume moreover that $\Theta$ is a finite set. Then $\mathcal{Q}_n$ converges strongly to $\mathcal{Q}$, i.e. 
$\lim_{n\to\infty}\Delta(\mathcal{Q}_{n}, \mathcal{Q})=0$.
\end{lemma}
{\it Proof.} 
By definition \ref{def.weak.convergence} we may assume that the 
phases of the vectors $|\psi_{\theta}(n)\rangle$ have been defined such that
$$
\lim_{n\to\infty} 
\langle \psi_{\theta}(n) | \psi_{\tau}(n)\rangle=
\langle \psi_{\theta}, \psi_{\tau}\rangle, \qquad
\forall \theta,\tau\in \Theta.
$$
Let $\Theta= \{1,\dots, k\}$ and define the (positive) Gram matrix 
$G^{(n)}_{i,j}:= \langle \psi_{i}(n) | \psi_{j}(n)\rangle$, and similarly 
$G_{i,j}:= \langle \psi_{i}| \psi_{j} \rangle$. 
The convergence of the Gram matrices implies that for $n$ large enough $G_n$ and $G$ have the same rank $r$. 

Let $P^{(n)}:=\sqrt{G^{(n)}}, P:=\sqrt{G}$ and let $\{ |e_{1}\rangle,\dots, |e_{k}\rangle\}$ be the standard orthonormal basis in $\mathbb{C}^{k}$. 
Then the vectors
$
|\widetilde{\psi}_{i}(n)\rangle := \sum_{l=1}^{k} P^{(n)}_{i,l}|e_{l}\rangle \in \mathbb C^k
$
have inner products 
$$
\langle\widetilde{\psi}_{i}(n)| \widetilde{\psi}_{j}(n) \rangle= ((P^{(n)})^{*}P^{(n)})_{i,j}= 
G^{(n)}_{i,j}=\langle \psi_{i}(n)| \psi_{j}(n)\rangle .
$$
Similarly the vectors 
$
|\widetilde{\psi}_{i}\rangle := \sum_{l=1}^{k} P_{i,l} |e_{l}\rangle 
$
have inner products 
$$
\langle\widetilde{\psi}_{i}|\widetilde{\psi}_{j}\rangle= 
(P^{*}P)_{i,j} = 
G_{i,j}  =\langle \psi_{i}| \psi_{j}\rangle. 
$$

Since two sets of vectors with the same Gram matrix are related by a single unitary matrix, we can map the models $\mathcal{Q}_n$ and $\mathcal{Q}$ into equivalent ones 
$\widetilde{\mathcal{Q}}_n$ and respectively $\widetilde{\mathcal{Q}}$, defined on the same representation space $\mathbb C^k$. By weak convergence we have that $G^{(n)}$ converges entrywise, and hence in norm to $G$. As the square root is a continuous function we get that $P^{(n)}$ converges to $P$ and hence $|\tilde\psi_{i}(n)\rangle \to |\tilde\psi_{i}\rangle$ for all $i$. Therefore 
$\lim_{n\to\infty}\Delta(\widetilde{\mathcal{Q}}_n, \widetilde{\mathcal{Q}}) = 0$, and the same holds for 
$\mathcal{Q}_n$ and $\mathcal{Q}$ due to the above stated equivalence.
 
\qed

We will now strengthen this results, by assuming that we have more control over the weak convergence. The result will be used in upgrading the local asymptotic normality result in the next section, from weak to strong convergence. 
\begin{lemma}\label{th.weak.implies.strong}
Let $\Theta$ be a compact subset of $\real^k$, and let $\mathcal{Q}_{n}$ and $\mathcal{Q}$ be as in definition \ref{def.weak.convergence}, with $\mathcal{Q}_{n}$ converging weakly to $\mathcal{Q}$. Suppose, in addition, that the convergence in \eqref{eq.weak.convergence} is uniform for $(\theta_1,\theta_2)\in \Theta$.
Then $\mathcal{Q}_{n}$ converges strongly to $\mathcal{Q}$. 
 \end{lemma}

The proof can be found in section \ref{sec.proof.weaktostrong}.
\section{Local asymptotic normality for quantum Markov chains}
\label{sec.lan}

We now proceed to the second main task outlined in the introduction: characterising the statistical properties of the quantum output state of the Markov chain. We will assume that the Markov dynamics is unknown, or more precisely that $V= V_\theta$ depends smoothly on some unknown parameter $\theta\in \Theta\subset \mathbb{R}$. The case of multidimensional parameters can be investigated along the same lines, but will not be discussed here. A practical goal is to estimate $\theta$ by performing measurements on the output and constructing an estimator $\hat\theta$ based on the measurement results. The restriction to output measurements is justified by the fact that in many experiments the system cannot be measured directly, and is particularly natural in the context of quantum control engineering \cite{Wiseman&Milburn}. This problem was studied in \cite{Guta} for a particular parametric family of discrete time quantum Markov chains, and the continuous time set-up was considered in \cite{Catana&Bouten&Guta}. The approach is based on tools of asymptotic statistics, which allows to distill the essential features of the problem such as asymptotic normality and Fisher information. 

In this paper we focus on the properties of the quantum statistical model rather than studying particular measurement strategies, in the spirit of the classic work on quantum Cram\'{e}r-Rao bound \cite{Belavkin,Holevo,Helstrom}. 
We first show that if the parameter $\theta$ is identifiable, the quantum statistical model 
$|\Psi_\theta(n)\rangle:= |\Psi_{V_\theta}(n)\rangle$ (along with three other models describing the output) can be approximated by a quantum Gaussian shift model, and its limiting quantum Fisher information (per unit of time) can be computed analytically. On the other hand, for un-identifiable parameters such as in conjugation with a local unitary, we show that the corresponding Fisher information is zero.

The convergence to the Gaussian model is the extension to Markov dynamics of the Holstein-Primakov convergence for spin coherent states discussed in section \ref{sec.i.i.d.lan}. This involves a rescaling of the parameter with the statistical uncertainty such that $\theta:= \theta_0 + u/\sqrt{n}$ with $u$ the unknown local parameter. As already explained this does not amount to an additional assumption since for large $n$ the `localisation' can be achieved by means of an adaptive estimation procedure. 

%

\subsection{Quantum statistical models for estimation of Markov processes}
\label{section.models}

We will analyse four (slightly) different definitions of the `output state' associated to a primitive isometry $V$, and show that in asymptotics they lead to the same limit model and quantum Fisher information.  
Using the notations from subsection \ref{notations}, we define the following states, where for the sake of clarity, the dependence on the initial system state and the isometry $V$ is indicated explicitly:
\begin{itemize}
\item[(a)] The joint (pure) state of system and output and the memory after $n$ iterations, for a given initial pure state $\varphi\in \hi$ of the system
$$
|\Psi_{V,\varphi}(n)\rangle:= V(n)|\varphi\rangle;
$$

\item[(b)] The output state after $n$ iterations, for a given initial pure state $\phi$ of the system
\begin{equation}\label{outputstate}
\rho_{V,\varphi}(n):={\rm tr}_{\hi}[|\Psi_{V,\varphi}(n)\rangle\langle \Psi_{V,\varphi}(n)|];
\end{equation}

\item[(c)] The output state  corresponding to the stationary regime
$$
\rho_V(n):={\rm tr}_{\hi}[V(n)\rho_{ss}V(n)^*];
$$
\item[(d)] The (un-normalized) conditional state
\begin{equation}\label{MPS}
|\psi_{V,\eta,\varphi}(n)\rangle:=\sum_{{\bf i}\in I^{(n)}} \langle \eta |{\bf K}_{\bf i}^{(n)}|\varphi\rangle |{\bf i}\rangle
\end{equation}
of the output after $n$ iterations, on the condition that the initial state of the memory is $|\varphi\rangle$, and the measurement of the projection $|\eta\rangle\langle\eta|$ on the memory after the $n$ iterations yields $1$. In case 
$|\eta\rangle$ or $|\varphi\rangle$ is one of the basis vectors $|e_j\rangle$, we replace it by the index $j$ so as to simplify the notation.
\end{itemize}

Here (c) is most natural from the operational point of view, since our basic physical assumption is that we only have access to the output in the stationary regime. The other states are, however, technically easier to handle; in particular, (d) and (a) have explicitly the form of a Matrix Product State.

Let us clarify the relations between the states $(a)-(d)$; this will be useful later. Let
$
\rho_{ss} = \sum_{i=1}^D \Lambda_i |e_i\rangle\langle e_i|
$
be the spectral decomposition of the stationary state $\rho_{ss}$. 
We can decompose the states (a) as
\begin{equation}\label{conditional}
|\Psi_{V,\varphi}(n) \rangle= \sum_{j=1}^D |e_j \rangle\otimes |\psi_{V,j, \varphi}(n) \rangle,
\end{equation}
and the state (b) is a mixture of states of type $(d)$
\begin{equation}\label{eq.rho.varphi}
\rho_{V,\varphi}(n) = \sum_{j=1}^D |\psi_{V,j,\varphi} (n)\rangle\langle \psi_{V,j,\varphi} (n)|.
\end{equation}

The stationary output  state (c) is the mixture of pure states of type (d) 
corresponding to the system starting in one of the eigenstates $|e_i\rangle$, and ending in a state $|e_j\rangle$
\begin{align}\label{eq.decomposition.stationary.state}
\rho_V(n)&=\sum_{i=1}^D \Lambda_{i} {\rm tr}_{\hi}[|V(n) e_i\rangle\langle V(n) e_i|]\nonumber\\
&=\sum_{i,j=1}^D \Lambda_{i} |\psi_{V,j,i}(n) \rangle\langle \psi_{V,j,i}(n) |\\
&=\sum_{i,j=1}^D \Lambda_{i}\Lambda_{j} |\Lambda_{j}^{-\frac 12}\psi_{V,j,i}(n) \rangle\langle \Lambda_{j}^{-\frac 12}\psi_{V,j,i}(n) |\nonumber.
\end{align}

\subsection{Main results}
In this section we present the second main result of the paper and discuss the interpretation of the associated quantum Fisher information. The proof is discussed in section \ref{proof.LAN}. For the notations we refer to section \ref{sec.markov.covariance}.

Let $\theta\mapsto V_\theta$, be a smooth family of isometries parametrized by the unknown parameter $\theta\in \real$, and let $K_{i,\theta}$ be the corresponding Kraus operators, cf. section \ref{sec.Markov.dynamics}. We write $\theta= \theta_0 + u/\sqrt{n}$, with local parameter $u$ and assume that $V:= V_{\theta_0}$ is a primitive isometry, with stationary state $\rho_{ss}$. Let $\Lambda_i$ denote its eigenvalues. 
We denote by $K_i:= K_{i, \theta_0}$ and by $\dot{K}_i$ the derivative of $K_{i,\theta}$ with respect to $\theta$, at $\theta_0$. To emphasise the dependence on the local parameter we denote  by $|\Psi_{u,\varphi}(n)\rangle$ the output state corresponding to $V_{\theta_0+ u/\sqrt{n}}$, and similarly for the other states introduced in the previous subsection: 
$ \rho_{u,\varphi}(n)$, $\rho_u(n)$ and $|\psi_{u,\eta,\varphi}(n)\rangle$. All the results in this section are based on the following Theorem \ref{perturbationtheorem}, together with the observation that the relevant scalar products can clearly be written in the form
\begin{align}
\langle \Psi_{u,\varphi}(n)|\Psi_{v,\varphi}(n)\rangle&=\langle \varphi |T^n_{u,v,n}(\id)\varphi\rangle,\label{scalar1}\\
\langle \psi_{u,\eta,\varphi}(n)|\psi_{v,\eta,\varphi}(n)\rangle &={\rm tr}[T_{u,v;n}^n(|\eta\rangle\langle \eta|)|\varphi\rangle\langle\varphi|]\label{scalar2},
\end{align} 
where
\begin{align*}
T_{u,v;n}:&\lh\to\lh, & T_{u,v;n}(X):=V_{\theta_0+u/\sqrt{n}}^*(X\otimes \id)V_{\theta_0+v/\sqrt{n}}.
\end{align*}
In order to simplify calculations we will additionally assume the following `gauge condition'
\begin{equation}\label{eq.gauge}
{\rm Im} \sum_{i=1}^k {\rm tr}[\rho_{ss}\dot{K}_{i}^* K_{i}]=0.
\end{equation}
It is straightforward to verify that \eqref{eq.gauge} can always be satisfied by replacing the isometries $V=V_\theta$ with $e^{ib\theta} V_\theta$, where $b$ is a suitable constant. Note that this transformation does not affect the actual statistical models given by the states 
$\Psi_{u,\phi}(n)$, $\rho_{u,\phi}(n)$, $\rho_u(n)$, and $\psi_{u,\eta,\phi}(n)$. It also does not change the equivalence class of the isometry $V_\theta$.
\begin{theorem} \label{perturbationtheorem}
There exist $F,a\in \real$, such that for each $X\in \mathcal{B(H)}$ and $C>0$ we have
\begin{equation}\label{convergence}
\lim_{n\rightarrow\infty} \sup_{|u|,|v|<C}T^n_{u,v;n}(X)={\rm tr}[\rho_{ss}X] e^{-F(u-v)^2/8- ia(u^2- v^2)}\id.
\end{equation}


The constant $F$ can be written in terms of the Markov variance of the `generator' $G^*:= i \dot{V}V^*$  
(cf. section \ref{sec.markov.covariance})
\begin{eqnarray}
F= 4 (G^*,G^*)_V &=&4 {\rm Tr}\left\{\rho_{ss} \mathcal{E}\left[
GG^* +  2{\rm Re} [ G (\id \otimes \mathcal{R}\circ\mathcal{E}(G^*) )]  \right]\right\}
\nonumber\\
&=& 4\sum_{i=1}^k \left[ {\rm Tr}\left[\rho_{ss}\dot{K}_i^*\dot{K}_i\right]  + 2{\rm Tr} \left[{\rm Im} ( K_i  \rho_{ss}  \dot{K}_i^* ) \cdot\mathcal{R}( {\rm Im} \sum_k\dot{K}_{i}^* K_i )\right]\right].\label{eq.fisher.2}
\end{eqnarray}

\end{theorem}

\begin{proof}
see section \ref{sec.proof.key}. 
\end{proof}

\subsubsection{LAN for the pure system and output state}


Let $\mathcal{Q}_{n,\varphi}:= \{ |\Psi_{u,\varphi} (n)\rangle: u\in \mathbb{R}\}$ be the quantum statistical model consisting of
of system and output state $|\Psi_{V,\varphi}(n)\rangle$ (cf. section \ref{section.models}),  with $V= V_{\theta_0+u/\sqrt{n}}$, and initial state $|\varphi\rangle$. We recall that
$$\mathcal{G}:= \left\{ |\sqrt{F/2}u\rangle \in  L^2(\real) \,:\, u\in \mathbb{R} \right\}$$ is a  one-dimensional (pure states) quantum  Gaussian model consisting of coherent states with mean $(\sqrt{F/2}u,0)$.

\begin{theorem}\label{th.qlan} 
The following convergence results hold:
\begin{itemize}
\item[1.] $\mathcal{Q}_{n,\varphi}$ converges weakly to $\mathcal{G}$; more precisely, there exists a real constant $a$ such that 
\begin{equation}\label{eq.qlan}
\lim_{n\to \infty} \left\langle \exp(-iau^2) \Psi_{u,\varphi}(n) |\exp(-iav^2) \Psi_{v,\varphi}(n)\right\rangle = \langle \sqrt{F/2} u | \sqrt{F/2}v\rangle =  \exp(- F(u-v)^2/8)
\end{equation}
where $\exp (- iau^2)$ is a physically irrelevant phase factor, and $F$ is the asymptotic quantum Fisher information per time unit, given by \eqref{eq.fisher.2}. 
\item[2.] Let $Q_{n,\varphi}^c, G^c$ be the restrictions of $\mathcal{Q}_{n,\varphi}$ and $\mathcal{G}$ to parameters in the bounded interval $u\in [-c,c] $.
Then $Q_{n,\varphi}^c$ converges strongly to $\mathcal{G}^c$, cf. definition \ref{def.strong.convergence}.
\end{itemize}
\end{theorem}
\begin{proof} The weak convergence follows immediately from Theorem \ref{perturbationtheorem} and \eqref{scalar1}.
The strong convergence follows from the same Theorem \ref{perturbationtheorem} which shows that the above convergence is uniform 
over $u,v$ in a bounded interval, and Lemma \ref{th.weak.implies.strong} which upgrades weak to strong convergence.

\end{proof}

\subsubsection{LAN for the pure conditional output state}

Since the system plus output state $|\Psi_{u,\varphi}(n)\rangle$ is pure (for a pure system initial state), its statistical properties are easier to understand, for which reason we have concentrated our analysis on this model so far. However, if we restrict ourselves to output measurements, it is natural to consider the output state alone, and its  Fisher information. Intuitively, since in the long time limit the system reaches stationarity, while the Fisher information of the output state grows linearly, one expects that all information is in the output and no loss of information (per time unit) is incurred by tracing out the system. 

We consider the output states $|\psi_{u,\eta,\varphi}(n)\rangle$, corresponding to an initial state $|\varphi\rangle$, and a final state 
$|\eta\rangle$ after $n$ iterations. Let $\mathcal{Q}_{n,\eta,\varphi}:= \{ \psi_{u,\eta,\varphi}(n)/\sqrt{N_n(u,\eta,\varphi)} : u\in \mathbb{R}\}$ be the corresponding statistical model, where $N_n(u,\eta,\varphi) = \| \psi_{u,\eta,\varphi}(n)\|^2$ is the normalisation constant.

\begin{theorem}\label{asymptotic_normality2} The following convergence results hold:
\begin{itemize}
\item[1.] $\mathcal{Q}_{n,\eta,\varphi}$ converges weakly to $\mathcal{G}$.
\item[2.] Let $Q_{n,\eta,\varphi}^c, \mathcal G^c$ be the restrictions of $\mathcal{Q}_{n,\eta,\varphi}$ and $\mathcal{G}$ to parameters in the bounded interval 
$u\in [-c,c] $. Then $Q_{n,\eta,\varphi}^c$ converges strongly to $\mathcal{G}^c$, cf. definition \ref{def.strong.convergence}.
\end{itemize}
\end{theorem}
\begin{proof}
It follows immediately from Theorem \ref{perturbationtheorem} and \eqref{scalar2} that
$$
\lim_{n\rightarrow\infty} 
 \langle \psi_{u,\eta,\varphi}(n)|\psi_{v,\eta,\varphi}(n)\rangle
= \langle \eta |\rho_{ss}\eta\rangle \langle \sqrt{F/2}u,\sqrt{F/2}v\rangle,
$$
which also implies that $\lim_{n\rightarrow \infty} N_n(u,\eta,\varphi) =\langle \eta |\rho_{ss}\eta\rangle$. This proves the weak convergence, and the strong convergence again follows from Theorem \ref{perturbationtheorem} and Lemma \ref{th.weak.implies.strong}.
\end{proof}
\subsubsection{LAN for mixed output state with known initial state $|\eta\rangle$}

In a more realistic setup, we do not have access to the system after $n$ iterations. 
Assuming that we know the initial state $|\varphi\rangle$, the output is described by the mixed state
statistical model  $\mathcal{O}_{n,\varphi}:= \{\rho_{u,\varphi}(n) : u\in \mathbb{R}\}$, and its restriction to local parameters 
$u\in [-c,c]$ is denoted $\mathcal{O}_{n,\varphi}^c$. Since our definition of weak convergence applies only to pure states models, we formulate local asymptotic normality as strong converge to the Gaussian model. 

\begin{theorem}\label{asymptotic_normality3}
The sequence of statistical models $\mathcal{O}_{n,\varphi}^c$ converges strongly to $\mathcal{G}^c$.
\end{theorem}
\begin{proof}
We use the decomposition \eqref{eq.rho.varphi} for $\rho_{u,\varphi}(n)$, with the eigenbasis of the stationary state $\rho_{ss}$ of $V$:
$$
\rho_{u,\varphi}(n) = \sum_{j=1}^D |\psi_{u,j,\varphi}(n)\rangle\langle \psi_{u,j,\varphi}(n)|.
$$
Let 
$$
\tilde{\mathcal{O}}_{n,\varphi} := \{ |\psi_{u,j,\varphi}(n)\rangle/ \sqrt{N_n(u,j,\varphi)} : (u, j)\in \mathbb{R}  \times \{1,\dots, D\}\}
$$
be the `extended' model consisting of all pure components $|\psi_{u,j,\varphi}(n)\rangle$, where the index $j$ is seen as a discrete parameter. 
Again by using Theorem \ref{perturbationtheorem} and \eqref{scalar2}, we see that the extended model converges weakly to the extended Gaussian model 
$$
\tilde{G} := \{ | \sqrt{F/2} u \rangle\otimes |j\rangle :  (u, j)\in \mathbb{R}  \times \{1,\dots, D\}\}
$$
where $| \sqrt{F/2} u \rangle \otimes |j\rangle$ denotes a copy of the coherent state $| \sqrt{F/2} u \rangle$ in the tensor product 
$L^2(\mathbb{R})\otimes \mathbb{C}^D$.
Moreover since the above convergence is uniform, we can again apply Lemma \ref{th.weak.implies.strong} to upgrade the weak convergence to strong convergence. This means that for each $c>0$ there exist quantum channels $\tilde{T}_n$ such that
\begin{equation}\label{complimit}
\lim_{n\to\infty} \sup_{|u|\leq c} \,\max_{j=1,\dots, D} 
\| N_n(u,j,\varphi)^{-1}\tilde{T}_n (|\psi_{u,j,\varphi}(n)\rangle \langle \psi_{u,j,\varphi}(n)| ) -  | \sqrt{F/2} u \rangle\langle \sqrt{F/2} u | \otimes|j\rangle\langle j|  \|_1 =0.
\end{equation}
Let $P$ be the channel taking partial trace over the label space $\mathbb{C}^D$. Then
$$
P\left(| \sqrt{F/2} u \rangle\langle \sqrt{F/2} u |  \otimes \sum_{j=1}^D\Lambda_j|j\rangle\langle j|\right)=| \sqrt{F/2} u \rangle\langle \sqrt{F/2} u |,
$$
so using the channel contractivity for $P$ we get
$$
\| T_n (\rho_{u,\varphi}(n)) - | \sqrt{F/2} u \rangle\langle \sqrt{F/2} u |     \|_1\leq \sum_{j=1}^D \left\| \tilde{T}_n (|\psi_{u,j,\varphi}(n)\rangle \langle \psi_{u,j,\varphi}(n)| ) -  \Lambda_j | \sqrt{F/2} u \rangle\langle \sqrt{F/2} u | \otimes|j\rangle\langle j|  \right\|_1,
$$
where $T_n = P \circ \tilde{T}_n$. Since $\lim_{n\rightarrow\infty} N_n(u,j,\varphi)=\Lambda_j$, it follows from \eqref{complimit} that
$$
\lim_{n\to\infty} \sup_{|u|\leq c}
\| T_n (\rho_{u,\varphi}(n)) - | \sqrt{F/2} u \rangle\langle \sqrt{F/2} u |     \|_1 =0.
$$
The existence of channels $S_n$ such that
$$
\lim_{n\to\infty} \sup_{|u|\leq c}
\| \rho_{u,\varphi}(n) - S_n(| \sqrt{F/2} u \rangle\langle \sqrt{F/2} u |)     \|_1 =0.
$$
is proved in a similar fashion.
\end{proof}
\subsubsection{LAN for the stationary output state}
We consider now the stationary output model $\mathcal{S}\mathcal{O}_n:= \{ \rho_u(n) : u\in \mathbb{R}\}$. For these states the decomposition \eqref{eq.decomposition.stationary.state} holds. We denote by $\mathcal{S}\mathcal{O}_n^c$ the restriction of $\mathcal{S}\mathcal{O}_n$ to parameters $u\in [-c,c]$, and we will show that $\mathcal{S}\mathcal{O}_n^c$ converges strongly to the (restricted) Gaussian model $\mathcal{G}^c$. This means that asymptotically, the stationary state is carries the same amount of (Fisher) information as the system-output state which is a priori more informative since the observer has access to both system and output. The two models are asymptotically equivalent with a pure state Gaussian model.

\begin{theorem}\label{asymptotic_normality} For any $c>0$, the stationary output state statistical model $\mathcal{S}\mathcal{O}_n^c$ converges strongly to $\mathcal{G}^c$.
\end{theorem}

\begin{proof} We use the decomposition \eqref{eq.decomposition.stationary.state} of $\rho_{u}(n)$:
$$
\rho_{u}(n) = \sum_{j=1}^D\Lambda_i|\psi_{u,j,i}(n)\rangle\langle \psi_{u,j,i}(n)|.
$$
Let 
$$
\widetilde{\mathcal{S}\mathcal{O}}_{n} := \{ |\psi_{j,i,u}(n)\rangle/ \sqrt{N_n(u,j,i)} : (u, j,i)\in \mathbb{R}  \times \{1,\dots, D\}^2\}
$$
be the `extended' model consisting of all pure components $|\psi_{u,j,i}(n)\rangle$, where the indices $j,i$ are seen as discrete parameters. Again we can use Theorem \ref{perturbationtheorem} and \eqref{scalar2} to conclude that this extended model converges weakly to the extended Gaussian model 
$$
\tilde{G} := \{ | \sqrt{F/2} u \rangle\otimes |j,i\rangle :  (u, j,i)\in \mathbb{R}  \times \{1,\dots, D\}^2\},
$$
where $| \sqrt{F/2} u \rangle \otimes |j,i\rangle$ denotes a copy of the coherent state $| \sqrt{F/2} u \rangle$ in the tensor product 
$L^2(\mathbb{R})\otimes \mathbb{C}^{D^2}$.
Moreover since the above convergence is uniform, we can apply Lemma \ref{th.weak.implies.strong} to upgrade the weak convergence to strong convergence. This means that there exist quantum channels $\tilde{T}_n$ such that
$$
\lim_{n\to\infty} \sup_{|u|\leq c} \,\max_{j,i=1,\dots, D} 
\| N_n(u,j,i)^{-1}\tilde{T}_n (|\psi_{j,i,u}(n)\rangle \langle \psi_{j,i,u}(n)| ) -  | \sqrt{F/2} u \rangle\langle \sqrt{F/2} u | \otimes|j,i\rangle\langle j,i|  \|_1 =0.
$$
Let $P$ be the channel taking  partial trace over the label space $\mathbb{C}^{D^2}$, and take $T_n = P \circ \tilde{T}_n$. Since
$$
P\left(| \sqrt{F/2} u \rangle\langle \sqrt{F/2} u |  \otimes \sum_{j,i=1}^D \Lambda_i\Lambda_j|j,i\rangle\langle j,i|\right)=| \sqrt{F/2} u \rangle\langle \sqrt{F/2} u |,
$$
we conclude, as in the proof of Theorem \ref{asymptotic_normality3}, that
$$
\lim_{n\to\infty} \sup_{|u|\leq c}
\| T_n (\rho_{u}(n)) - | \sqrt{F/2} u \rangle\langle \sqrt{F/2} u |     \|_1 =0.
$$
Again the other limit needed for the strong convergence is proved in a similar fashion.
\end{proof}

\subsection{Unidentifiable parameters have zero Fisher information}

If the Fisher information $F$ at $\theta_0$ is strictly positive, the above theorems \ref{th.qlan}, \ref{asymptotic_normality2}, \ref{asymptotic_normality3}, and \ref{asymptotic_normality} imply that $\theta$ can be estimated with standard mean square error rate $1/n$ in the neighbourhood of $\theta_0$, using any of the models considered.
 
On the other hand,
if the isometries $V_\theta$ in the family are all equivalent in the sense of Definition \ref{equiv_def}, the associated \emph{output} quantum Fisher information for such parameters must be strictly zero for outputs at arbitrary time (i.e. length $n$). However, for a given initial state $\varphi$ of the system, the \emph{joint} system plus output state \emph{can} depend on the `hidden' unitary conjugation and therefore may carry some information about it. Nevertheless, this information is not extensive in time, since the dynamics reaches stationarity in a finite time, and hence we would not expect it to contribute to the limiting Fisher information. Indeed, according to the results of the preceding section, the output model has the same limiting Fisher information per unit of time as the joint system plus output state.

The next corollary provides a consistency check, showing that any family of equivalent isometries has indeed zero limiting Fisher information for all of the statistical models considered above. The proof of this corollary can be found in Section \ref{proof.cor.Fisher.zero}.

\begin{corollary}\label{example_cor}\label{cor.Fisher.zero} Let $V_\theta$, $\theta\in \Theta\subset \real$ be an analytic family of isometries, all belonging to the same equivalence class (c.f. Definition \ref{equiv_def}). Then the quantum Fisher information $F= F(\theta_0)$ associated to $V_{\theta_0}$ is equal to zero for all $\theta_0$.
\end{corollary}


\section{Proofs} 
\label{proof.LAN}
In this section we collect the proofs of various results of the paper.

\subsection{Proof of Lemma \ref{th.weak.implies.strong}} \label{sec.proof.weaktostrong}

Denote $f_n(\theta_1,\theta_2) =|\langle \psi_{\theta_1}(n) |\psi_{\theta_2}(n) \rangle |^{2}$, and $f(\theta_1,\theta_2)=|\langle \psi_{\theta_1}|\psi_{\theta_2}\rangle|^2$ for all $\theta_1,\theta_2\in \Theta$. 
Let us also denote the states 
$
\rho_{\theta}(n):= 
\left\vert \psi_{\theta}(n) \right\rangle 
\left\langle \psi_{\theta}(n) \right\vert
$ 
and  
$
\rho_{\theta}:= 
\left\vert \psi_{\theta} \right\rangle 
\left\langle \psi_{\theta} \right\vert.
$ 

We prove that $\lim_{n\rightarrow\infty}\delta(\mathcal Q_n,\mathcal Q)=0$. The other limit $\lim_{n\rightarrow\infty}\delta(\mathcal Q,\mathcal Q_n)=0$ can be proved in a similar manner.

Let $\epsilon>0$. According to the definition of strong convergence, we have to show that there exists an $n_0$ such that for each $n\geq n_0$ we have
$$
\sup_{\theta\in \Theta} \|T_n(\rho_\theta(n)-\rho_\theta\|_1<\epsilon
$$
for some channels $T_n$ (which may of course also depend on $\epsilon$).

Now by the assumption, the sequence $(f_n)$ of functions converges to $f$ uniformly on $\Theta\times\Theta$. Consequently, $f$ is continuous and hence uniformly continuous due to compactness of $\Theta$, so a simple $\epsilon/3$-argument shows that $(f_n)_{n\geq n_1}$ is a \emph{uniformly equicontinuous} family for some $n_1\in \mathbb N$. This means that there exists a $\delta>0$ such that
$$
\sup_{n\geq n_1} |f_n(\theta_1,\theta_2)-f_n(\theta_1',\theta_2')|<\epsilon^2/36
$$
whenever $|\theta_1-\theta_1'|<\delta$ and $|\theta_2-\theta_2'|<\delta$. Since $f_n(\theta,\theta)=1$, we get, in particular, that
\begin{equation}\label{uniform}
\|\rho_{\theta}(n) -\rho_{\tau}(n) \|_1=2\sqrt{1-|\langle \psi_\theta(n) | \psi_\tau(n)\rangle |^{2}}
=2\sqrt{|f_n(\theta,\tau)-f_n(\theta,\theta)|}<\epsilon/3
\end{equation}
whenever $n\geq n_1$ and $|\theta-\tau|<\delta$. Moreover, $(f_n)$ converges to $f$ pointwise, so
\begin{equation}\label{uniform2}
\|\rho_{\theta} -\rho_{\tau}\|_1=2\sqrt{|f(\theta,\tau)-f(\theta,\theta)|}=\lim_{n\rightarrow\infty} 2\sqrt{|f_n(\theta,\tau)-f_n(\theta,\theta)|}\leq \epsilon/3
\end{equation}
whenever $|\theta-\tau|<\delta$.

Since $\Theta$ is compact, the exists a finite set $I_\delta\subset \Theta$ such that
\begin{equation}\label{cover}
\Theta \subset\bigcup_{\tau\in I_\delta} B_\delta(\tau),
\end{equation}
where $B_\delta(\tau)$ denotes the open ball of radius $\delta$ centered at $\tau$. The restrictions of $\mathcal Q_n$ and $\mathcal Q$ to $I_\delta$ are now finite state models. In this case weak convergence is equivalent with strong convergence as shown in Lemma \ref{lemma.weak.stronng.equivalence}, so
there exist channels $T^{\delta}_{n}$ such that
\begin{equation*}
\lim_{n \to \infty}\sup_{\tau\in I_\delta}
\| T_{n}^{\delta} \rho_{\tau}(n) -\rho_{\tau}\|_1 =0.
\end{equation*}
Consequently, we can find an $n_0\geq n_1$ such that
\begin{equation}\label{finiteconvergence}
\| T_{n}^{\delta} \rho_{\tau}(n) -\rho_{\tau}\|_1<\epsilon/3 \quad \text{ for all }n\geq n_0, \, \tau\in I_\delta.
\end{equation}
Now for arbitrary $\theta\in \Theta$, there exists a $\tau\in I_\delta$ such that $|\theta-\tau|<\delta$ due to \eqref{cover}. With such a $\tau$, the inequalities \eqref{uniform}, \eqref{uniform2}, and \eqref{finiteconvergence} hold if $n\geq n_0$. Hence, using the fact that every channel has norm one, we get
$$
\| T_{n}^{\delta} \rho_{\theta}(n) -\rho_{\theta}\|_1 \leq 
\| T_{n}^{\delta}  ( \rho_{\theta}(n) -\rho_{\tau}(n) ) \|_1 +
\| T_{n}^{\delta}  \rho_{\tau}(n) - \rho_{\tau}\|_1+
\|\rho_{\tau}- \rho_{\theta}\|_1<\epsilon,
$$
whenever $n\geq n_0$. This completes the proof.
\qed

\subsection{Proof of Theorem \ref{perturbationtheorem}}\label{sec.proof.key}
By analyticity, we can expand the Kraus operators using the Taylor series at $\theta_0$:
$$
K_{i,\theta}=K_i+(\theta-\theta_0)\dot K_i +\frac 12(\theta-\theta_0)^2\ddot K_i+O((\theta-\theta_0)^3).
$$
Hence, for a fixed $C>0$, we have
$$
K_{i,\theta_0+u/\sqrt n}=K_i+\frac{1}{\sqrt n}u\dot K_i +\frac 12\frac{1}{n}u^2\ddot K_i+O(n^{-\frac 32}).
$$
uniformly for $|u|<C$, in the sense that the error term can be bounded above by a constant times $n^{-\frac 32}$. Using the fact that
$$
T_{u,v;n}(X)=\sum_i K^*_{i,\theta_0+u\sqrt{n}}XK_{i,\theta_0+v\sqrt{n}}
$$
it then follows that, again uniformly for $|u|<C$, $|v|<C$, we have
\begin{equation}\label{Texpansion}
T_{u,v;n} = T + \frac{1}{\sqrt{n}} T_{1,u,v} + \frac{1}{n} T_{2,u,v}+ O(n^{-3/2}),
\end{equation}
where $T$ is the channel associated to $V$, and
\begin{eqnarray}
T_{1,u,v}(X)&=& \sum_{i=1}^k \left\{u \dot{K}^*_i XK_i + v K_i^* X \dot{K}_i  
\right\}\nonumber\\
&=& u \dot{V}^* (X\otimes \id)  V + vV^* (X\otimes \id)  \dot{V}\label{eq.t1}\\
T_{2,u,v}(X)&=&\sum_{i=1}^k \left\{ \frac{u^2}{2} \Ddot{K}_i^* X K_i + \frac{v^2}{2} K_i^* X\Ddot{K}_i + uv \dot{K}_i^* X \dot{K}_i \right\}
\nonumber\\
&=& \frac{u^2}{2} \Ddot{V}^* (X\otimes \id) V + \frac{v^2}{2} V^* (X\otimes\id) \Ddot{V}  + uv \dot{V}^* (X\otimes\id) \dot{V}.
\end{eqnarray}
Using the fact that $\sum_i \dot K_i^*K_i=-\sum_i K_i^*\dot K_i$ by normalisation, we see that \eqref{eq.gauge} implies
\begin{equation}\label{ortho}
(T_{1,u,v}(\id), \id)_{\theta_0}= i(u-v) \sum_{i=1}^k{\rm Im}( \dot{K}_i^* K_i , \id )_{\theta_0} =0,
\end{equation}
where $(A,B)_{\theta_0}:= {\rm Tr}(\rho_{ss} A^*B)$ defines an inner product on $\mathcal{B(H)}$ because $(A,A)_{\theta_0}=0$ implies ${\rm tr}[\rho_{ss}A^*A]=0$ due to $\rho_{ss}$ being full rank by the primitivity of $T$. Now \eqref{ortho} means that $T_{1,u,v}(\id)$ belongs to $\mathcal{D}_1:=\{X: {\rm Tr}( \rho_{ss} X)=0\}=({\rm Id}- T)(\mathcal{B(H)})$, which is the orthogonal complement of $\mathcal{D}_0:= \mathbb{C}\id=\ker({\rm Id}- T)$ with respect to the inner product $(\cdot,\cdot)_{\theta_0}$. Hence, the operator $\mathcal{R}(T_{1,u,v}(\id)) $ is well defined, where $\mathcal{R}$ is the inverse of the restriction of $({\rm Id}- T)$ to the invariant linear space $\mathcal{D}_1$ where that operator is bijective.
This condition, together with the primitivity assumption for $T$, and the fact that each $T_{u,v;n}$ is a contraction in the operator norm, makes up the hypothesis of a second order perturbation theorem (Theorem 2) in \cite{Guta}. Consequently, $\lim_{n\to\infty}\|T^n_{ u, v;n }(\id) - \exp(\lambda(u,v)) \id\|=0$ holds for each $u,v$, where $\lambda(u,v)$ is given by
\begin{align}\label{eq.lambda.uv}
\lambda(u,v) &= (\id,T_{2,u,v}(\id))_{\theta_0}+ (\id,T_{1,u,v}\circ \mathcal{R}\circ T_{1,u,v}(\id))_{\theta_0}.
\end{align}
Moreover, since the perturbative expansion \eqref{Texpansion} (in powers of $n^{-\frac 12}$) is uniform for $|u|,|v|<C$ for any given $C>0$, an inspection of the proof of the above mentioned theorem readily shows that, in fact,
\begin{equation}\label{general_convergence}
\lim_{n\to\infty} \sup_{|u|,|v|\leq C} \|T^n_{ u, v;n }(\id) - \exp(\lambda(u,v)) \id\|=0.
\end{equation}
From \eqref{Texpansion} we also directly get the bound
\begin{equation}\label{taylorbound}
\sup_{|u|,|v|\leq C}\|T_{u,v;n}-T\|\leq \sup_{|u|,|v|\leq C}\|T_{1,u,v}\| n^{-\frac 12}\leq {\rm const.} n^{-\frac 12}.
\end{equation}
As before we now write the spectral decomposition of $\rho_{ss} $ as
$$
\rho_{ss} = \sum_{i=1}^{\dim\hi} \Lambda_{i} |e_{i}\rangle\langle e_{i}|
$$
and we proceed to show the following statement 
\begin{equation}\label{matrix_convergence}
\lim_{n\to\infty} \sup_{|u|,|v|\leq C}\|T^n_{ u, v;n }(|e_j\rangle\langle e_i|) -e^{\lambda(u,v)} \Lambda_i\delta_{ij}\id\|=0.
\end{equation}
Let $k\in \nat$, $k<n$. Since $T_{u,v;n}$ is a \emph{contraction} in the operator norm, we can estimate
\begin{eqnarray*}
\|T_{u,v;n}^n (| e_j \rangle\langle e_i | ) - e^{\lambda(u,v)}\Lambda_i\delta_{i,j}  \id \| &\leq& 
\| T_{u,v;n}^{n-k} T^k (|e_j \rangle\langle e_i |)  - e^{\lambda(u,v)}\Lambda_i\delta_{i,j} \id\| \\
&+& 
\| T_{u,v;n}^k(|e_j \rangle\langle e_i |) - T^k (|e_j \rangle\langle e_i |)  \|\\
&\leq&
\Lambda_{i}\delta_{i,j} \|T_{u,v;n}^{n-k} (\id ) -e^{\lambda(u,v)}\id\|\\
&+& \|T^k(|e_j \rangle\langle e_i |) - \delta_{i,j} \Lambda_{i}\id\|\\
&+&
\| T_{u,v;n}^k(|e_j \rangle\langle e_i |) - T^k (|e_j \rangle\langle e_i |)  \|\\
&\leq&
\Lambda_{i}\delta_{i,j} (\|\id-T^k_{u,v,n}(\id)\|+\|T_{u,v;n}^{n} (\id ) -e^{\lambda(u,v)}\id\|)\\
&+& \|T^k(|e_j \rangle\langle e_i |) - \delta_{i,j} \Lambda_{i}\id\|\\
&+&
\| T_{u,v;n}^k(|e_j \rangle\langle e_i |) - T^k (|e_j \rangle\langle e_i |)  \|.
\end{eqnarray*}
For an arbitrary $\delta>0$, we first choose $k$ to be large enough so that $\|T^k(|e_j \rangle\langle e_i |) - \delta_{i,j} \Lambda_{i}\id\|<\delta/2$. This is possible because $T_0$ is primitive. With this fixed $k$, we have
\begin{align*}
\|\id-T^k_{u,v,n}(\id)\| &= \|T^k(\id)-T^k_{u,v,n}(\id)\|\leq k\|T-T_{u,v,n}\|,\\
\| T_{u,v;n}^k(|e_j \rangle\langle e_i |) - T^k (|e_j \rangle\langle e_i |)  \| &\leq k\|T-T_{u,v,n}\|.
\end{align*}
Hence, using \eqref{general_convergence}, and the bound \eqref{taylorbound}, we obtain
$$\limsup_{n\rightarrow\infty } \sup_{|u|,|v|\leq C}\|T_{u,v;n}^n (| e_j \rangle\langle e_i | ) - e^{\lambda(u,v)}\Lambda_i\delta_{i,j}  \id \|<\delta,$$ so \eqref{matrix_convergence} holds.

It remains to show that 
$$
\lambda(u,v) = - \frac{F(u-v)^2}{8} - ia(u^2- v^2)
$$
with $F$ given by \eqref{eq.fisher.2}. For this we compute now the two terms in \eqref{eq.lambda.uv}. Firstly
\begin{eqnarray*}
(\id,T_{2,u,v}(\id))_{\theta_0} &=& 
\frac{u^2+v^2}{2}\sum_{i=1}^k 
 {\rm Re} (\Ddot{K}_i, K_i)_{\theta_0} 
+ 
i \frac{u^2-v^2}{2} {\rm Im} (\Ddot{K}_i, K_i)_{\theta_0} \\
&+& uv \sum_{i=1}^k(\dot{K}_i ,\dot{K}_i )_{\theta_0}.
\end{eqnarray*}
Note that $\sum_{i=1}^k K_{i;\theta}^*K_{i;\theta}=\id$ implies
$
 {\rm Re} (\Ddot{K}_i, K_i)_{\theta_0} = - (\dot{K}_i, \dot{K}_i)_{\theta_0}
$
such that 
$$
(\id,T_{2,u,v}(\id))_{\theta_0} =-\frac{(u-v)^2}{2} 
  \sum_{i=1}^k (\dot{K}_i, \dot{K}_i)_{\theta_0}  + i\frac{u^2-v^2}{2} \sum_{i=1}^k {\rm Im} (\Ddot{K}_i, K_i)_{\theta_0}  .
$$
We pass now to the second term in \eqref{eq.lambda.uv}. By differentiating $\sum_{i=1}^k K_{i;\theta}^*K_{i;\theta}=\id$ we get 
$\sum_{i=1}^k(\dot{K}_i^* K_i+K_i^*\dot K_i )=0$ so that   
\begin{equation}\label{T1eq}
T_{1,u,v}(\id) = (u-v) \sum_{i=1}^k  \dot{K}_i^* K_i =i(u-v)( -i\dot{V}^*V).
\end{equation}
This implies 
\begin{equation*}
(\id,T_{1,u,v} \circ \mathcal{R}\circ T_{1,u,v}(\id))_{\theta_0}
=i(u-v) {\rm Tr}\left[ T_{1,u,v}^*(\rho_{ss} )  S\right],
\end{equation*}
where $S:= \mathcal{R} (-i\dot{V}^*V )$. 
From \eqref{eq.t1} we get
$$
T_{1,u,v}^*(\rho_{ss} ) = (u+v) {\rm Re} ( K_i \rho_{ss}  \dot{K}_i^* ) + i(u-v) {\rm Im} ( K_i \rho_{ss}  \dot{K}_i^* )
$$
which we insert in the previous formula to obtain
\begin{eqnarray*}
(\id,T_{1,u,v} \circ \mathcal R\circ T_{1,u,v}(\id))_{\theta_0}=&-&
(u-v)^2  \sum_{i=1}^k {\rm Tr} ({\rm Im} ( K_i \rho_{ss}  \dot{K}_i^* ) S )\\
&+&i (u^2-v^2)  \sum_{i=1}^k {\rm Tr} ({\rm Re} ( K_i \rho_{ss}  \dot{K}_i^* )S)
\end{eqnarray*}

Finally, putting together the two terms we obtain
\begin{eqnarray*}
\lambda(u,v) =& -& \frac{(u-v)^2}{2}  \sum_{i=1}^k 
 \left[ (\dot{K}_i, \dot{K}_i)_{\theta_0}  +2 {\rm Tr} ({\rm Im} (K_i \rho_{ss}  \dot{K}_i^* ) S )\right] \\
 &+& 
 i\frac{u^2-v^2}{2}  \sum_{i=1}^k \left[ 
 {\rm Im} (\Ddot{K}_i, K_i)_{\theta_0}+2\sum_{i=1}^k {\rm Tr} ({\rm Re} ( K_i \rho_{ss}  \dot{K}_i^* )S)\right]\\
 &=&
- \frac{(u-v)^2 F}{8} -  ia(u^2-v^2),
\end{eqnarray*}
where $a$ is an unimportant constant, and $F$ is the quantum Fisher information given by 
\begin{equation}\label{eq.qfisher}
F=  4\sum_{i=1}^k \left[ ({\rm Tr}[\rho_{ss}\dot{K}_i^*\dot{K}_i]  + 2{\rm Tr} ({\rm Im} ( K_i  \rho_{ss}  \dot{K}_i^* ) S)\right].
\end{equation}

We will now show that $F$ can be expressed as a four times the variance of the (Markov) `generator' $G^*:= i \dot{V}V^*$, in analogy to the case of a unitary family of pure states. We have
$$
V^* GV= -i V^*V\dot{V}^*V= -i \dot{V}^*V = iV^*\dot V=V^*G^*V,
$$
because $\dot V^*V+V^*\dot V=0$. In particular, $V^*GV$ is selfadjoint. We can now write
$S= \mathcal{R}( V^* G^*V)$, which is also selfadjoint because ${\rm Id} -T$, and hence also $\mathcal R$, commutes with the map $A\mapsto A^*$. Now
\begin{eqnarray*}
\sum_{i=1}^k{\rm Tr}( {\rm Im} (K_i \rho_{ss} \dot{K}_i^*) S )&=&
\sum_{i=1}^k{\rm Im}\, {\rm Tr}(K_i \rho_{ss} \dot{K}_i^* S )\\
&=&{\rm Im}\, {\rm Tr}(\rho_{ss}\dot V^*(\id\otimes S)V )\\
&=&{\rm Im}\, {\rm Tr}(\rho_{ss}V^*iG(\id\otimes S)V)\\
&=&{\rm Re}\, {\rm Tr}[\rho_{ss}\mathcal E(G(\id\otimes\mathcal R\circ\mathcal E(G^*)))].
\end{eqnarray*}
On the other hand
\begin{eqnarray*}
\sum_{i=1}^k{\rm Tr}[\rho_{ss}\dot{K}_i^*\dot{K}_i]    
 &=& 
 {\rm Tr}(\rho_{ss}V^*V \dot{V}^* \dot{V}V^*V) =  {\rm Tr}( \rho_{ss} \mathcal E(GG^*)).
\end{eqnarray*}
Therefore
$$
F= 4 (G^*,G^*)_V
=4 {\rm Tr}\left\{ \rho_{ss} \mathcal{E}\left[ GG^* +  2{\rm Re} [ G (\id\otimes \mathcal{R}\circ\mathcal{E}(G^*) ] \right]\right\}.
$$

\subsection{Proof of Corollary \ref{cor.Fisher.zero}}
\label{proof.cor.Fisher.zero}
\begin{proof} Fix $\theta_0\in \real$, and denote $V=V_{\theta_0}$, with also other notations as before. According to Theorem \ref{characterization}, for each $\theta\in \Theta$ there exists a unitary operator $\tilde U_\theta$ on $\mathcal H$, and a phase factor $e^{i f(\theta)}$ such that
\begin{equation}\label{unitary_family}
V_\theta= e^{if(\theta)}( \tilde U_\theta \otimes \id ) V \tilde U_\theta^*.
\end{equation}
Since the statistical models associated to $V_\theta$ are not affected by multiplication by a phase factor, we may assume without loss of generality that $f=0$. Now we have $V=(\tilde U_{\theta_0} \otimes \id ) V \tilde U_{\theta_0}^*$, which implies that $K_i\tilde U_{\theta_0}=\tilde U_{\theta_0}K_i$ for each Kraus operator $K_i$. Since $V$ is a primitive isometry, it follows that for a large enough $n$, the Kraus operators ${\bf K}^{(n)}_{\bf i}$ span the whole space $\mathcal{B(H)}$ (see e.g. \cite{Wolf}). Hence $X=\tilde U_{\theta_0}X\tilde U_{\theta_0}^*$ for all $X\in \mathcal{B(H)}$. From this it follows that $\tilde U_{\theta_0} =e^{ig(\theta_0)}\id$, where $e^{ig(\theta)}:=\langle 1|\tilde U_\theta|1\rangle$. Similarly, since $\theta\mapsto V_\theta$ is analytic, it follows that $\theta\mapsto \tilde U_\theta K_i \tilde U_\theta^*$ is analytic for all Kraus operators $K_i$, so again by primitivity, $\theta\mapsto \tilde U_\theta X\tilde U_\theta^*$ is analytic for all $X\in \mathcal{B(H)}$.

Hence, the unitary family $\theta\mapsto U_\theta:=e^{-ig(\theta)}\tilde U_\theta$ is analytic, satisfies $U_{\theta_0}=\id$, and we have $V_\theta=(U_\theta \otimes \id ) V U_\theta^*$.

Since $U_\theta$ is an analytic family, we may differentiate; we let $\dot U$ denote the derivative of $U_\theta$ at $\theta=\theta_0$, and put $$H:=-i\dot U^*.$$ Since $U_\theta^*U_\theta =\id$, and $U_0=\id$, we have
$\dot U^*+\dot U=0$, which implies $H=-i\dot U^*=i\dot U=H^*$, i.e. $H$ is selfadjoint. Hence,
\begin{equation}\label{dV}
i\dot V =i(\dot U\otimes \id)V+iV\dot U^*=(\id\otimes H)V-VH.
\end{equation}

We first verify that $(T_{1,u,v}(\id), \id)_{\theta_0}=0$ so no phase adjustment is needed for the Kraus operators. Indeed, using equation \eqref{T1eq} from the preceding proof, we get
\begin{align*}
(\id,T_{1,u,v}(\id))_{\theta_0}&={\rm tr}[\rho_{ss}T_{1,u,v}(\id)] =i(u-v){\rm tr}[\rho_{ss}(-i\dot V^*V)]\\
&=i(u-v){\rm tr}[\rho_{ss}(V^*(\id\otimes H)V-HV^*V)]=i(u-v){\rm tr}[\rho_{ss}(T(H)-H)]=0,
\end{align*}
because $\rho_{ss}$ is the stationary state.
In order to find the Fisher information we compute
\begin{align*}
G^*&=i\dot VV^*=(\id\otimes H)VV^*-VHV^*,\\
\mathcal E(G^*) &=V^*G^*V =V^*(\id\otimes H)V-VH=T(H)-H =-({\rm Id}-T)(H),\\
\mathcal R\circ \mathcal E(G^*)&=-H,\\
\mathcal E(G(\id\otimes\mathcal{R}\circ\mathcal{E}(G^*))&=
-\mathcal E(G(\id\otimes H))=-V^*(VV^*(\id\otimes H)-VHV^*)(\id\otimes H)V\\
&=-T(H^2)+HT(H),\\
\mathcal E\left(2{\rm Re}((G(\id\otimes\mathcal{R}\circ\mathcal{E}(G^*))\right)&=
2{\rm Re} \mathcal E(G(\id\otimes\mathcal{R}\circ\mathcal{E}(G^*))=2{\rm Re}(-T(H^2)+HT(H))\\
&=-2T(H^2)+HT(H)+T(H)H,\\
\mathcal E(GG^*) &=V^*(VV^*(\id\otimes H)-VHV^*)((\id\otimes H)VV^*-VHV^*)V\\
&=(V^*(\id\otimes H)-HV^*)((\id\otimes H)V-VH)\\
&=V^*(H^2\otimes\id)V-HV^*(\id\otimes H)V-V^*(\id\otimes H)VH+H^2\\
&=T(H^2)-HT(H)-T(H)H+H^2.
\end{align*}
Hence,
\begin{align*}
F&=4 {\rm tr}\left\{ \rho_{ss} \mathcal{E}\left[ GG^* +  2{\rm Re} [ G (\id\otimes \mathcal{R}\circ\mathcal{E}(G^*) ] \right]\right\}\\
&=4{\rm tr}[\rho_{ss}(-T(H^2)+H^2)]=0,
\end{align*}
because $\rho_{ss}$ is the stationary state.
%
\end{proof}

\section{Conclusions}

We have shown that the dynamics of an ergodic quantum Markov chain can be identified up to a `change of coordinates' transformation consisting of a phase multiplication and a unitary conjugation of the Kraus operators. This result can be seen as a quantum analogue of the classical result by Petrie \cite{Petrie} on hidden Markov chains. An interesting open problem is to find an explicit algorithm for computing the equivalence class from the output state.

We have further shown that the output state, which is a `purely generated' finitely correlated state \cite{FannesNachtergaeleWerner}, satisfies the local asymptotic normality property, and we have provided an explicit expression of the limiting quantum Fisher information, which can be interpreted as a Markov variance of a certain generator, in analogy to the case of unitary rotation families of quantum state. This analysis can be pushed forward by analysing the statistical properties of certain measurements such as simple repeated measurements on the output \cite{Guta}. In general such measurements do not achieve the quantum Fisher information and it remains an open question to characterise the optimal measurement procedure. This is related to another open problem, that of establishing a general quantum 
Central Limit for the fluctuation operators $\mathbb{F}_n(X)$ defined in  section \ref{sec.markov.covariance}.

More generally, we conjecture that local asymptotic normality holds for Markov chains with mixed i.i.d. input states. On the other hand, Markov chains with several ergodic components may fail to satisfy local asymptotic normality, and even exhibit `Heisenberg scaling', with relevance for quantum metrology applications.   

\

\noindent {\bf Acknowledgment.} This work was supported by the EPSRC project EP/J009776/1.

%
%



\end{document}